\documentclass{asl}

\usepackage{times}
\usepackage{amssymb}
\usepackage{prooftree}
\usepackage{latexsym}
\usepackage{stmaryrd}
\usepackage{graphicx}
\usepackage{mathabx}

\usepackage[raiselinks=false,plainpages=false,pdfpagelabels,naturalnames=true,colorlinks=true,citecolor=blue,urlcolor=blue,linkcolor=blue,bookmarksopen=true,backref,hyperindex]{hyperref}

\usepackage{mathpazo}
\usepackage{tikz}

\usetikzlibrary{fit}
\usetikzlibrary{shapes.misc}

\definecolor{darkblue}{rgb}{0,0,0.4}

\hypersetup{colorlinks,linkcolor=darkblue}

\newtheorem{definition}{Definition}[section]

\newtheorem{lemma}[definition]{Lemma}
\newtheorem{theorem}[definition]{Theorem}
\newtheorem{remark}[definition]{Remark}
\newtheorem{corollary}[definition]{Corollary}

\newtheorem{question}[definition]{Question}

\def\NN{\mathbb{N}}

\def\BB{\mathbb{B}}
\newcommand{\compact}{{\sf compact}}
\newcommand{\discrete}{{\sf discrete}}
\newcommand{\Type}{\mathcal{T}}

\newcommand{\overwrite}{\,@\,}

\newcommand{\PAomega}{{\rm PA}^\omega}
\newcommand{\HAomega}{{\rm HA}^\omega}

\newcommand{\ttimesTensor}{\otimes}
\newcommand{\ttimes}{\ttimesTensor}
\newcommand{\ttimesD}{\ttimesTensor_d}
\newcommand{\mttimes}{\,\tilde{\ttimesTensor}\,}
\newcommand{\mttimesD}{\,\tilde{\ttimesTensor}_d\,}

\newcommand{\EPQ}{{\sf EPQ}}	% iterated product of quantifiers
\newcommand{\EPS}{{\sf EPS}}	% conditional product of sel. fcts.
\newcommand{\IPS}{{\sf IPS}}		% iterated product of selection functions
\newcommand{\ips}{{\sf ips}}
\newcommand{\ipq}{{\sf ipq}}
\newcommand{\eps}{{\sf eps}}
\newcommand{\epq}{{\sf epq}}

\newcommand{\isp}{{\sf mbr}} 		% iterated skewed product
\newcommand{\ISP}{{\sf MBR}}

\newcommand{\BBC}{{\sf bbc}}
\newcommand{\BR}{{\sf BR}}
\newcommand{\SBR}{{\sf SBR}}
\newcommand{\mbr}{{\sf mbr}}
\newcommand{\MBR}{{\sf MBR}}
\newcommand{\length}{l}

\newcommand{\Rec}{{\sf R}}

\newcommand{\NUS}{J\textup{-}{\sf shift}}
\newcommand{\KSHIFT}{K\textup{-}{\sf shift}}

\newcommand{\CA}{{\rm CA}}
\newcommand{\AC}{{\rm AC}}
\newcommand{\DNS}{{\rm DNS}}
\newcommand{\CONT}{{\rm CONT}}
\newcommand{\SPEC}{{\rm SPEC}}
\newcommand{\BI}{{\rm BI}}

\newcommand{\fdefin}{=}
\newcommand{\mr}{\,{\sf mr}\,}
\newcommand{\eqleft}[1]{\begin{itemize} \item[] $#1$ \end{itemize}} 
\newcommand{\cZero}{{\bf 0}}
\newcommand{\selEmb}[1]{\overline{#1}}
\newcommand{\initSeg}[2]{[#1](#2)}
\newcommand{\initSegZ}[2]{\overline{#1, #2}}
\newcommand{\inj}{{\sf inj}}

\newcommand{\pair}[1]{\langle #1 \rangle} 

\newcommand{\cTrue}{{\sf tt}}
\newcommand{\cFalse}{{\sf ff}}

\newcommand{\cZr}[1]{\cZero^{#1}}
\newcommand{\lempty}{\langle\,\rangle}

\begin{document}

\title{Bar Recursion and Products of Selection Functions}

\author{Mart\'\i n Escard\'o and Paulo Oliva}

\date{{\sc Preprint, \today}}

\maketitle

\begin{abstract} We show how two iterated products of selection functions can both be used in conjunction with system $T$ to interpret, via the dialectica interpretation and modified realizability, full classical analysis. We also show that one iterated product is equivalent over system $T$ to Spector's bar recursion, whereas the other is $T$-equivalent to modified bar recursion. Modified bar recursion itself is shown to arise directly from the iteration of a different binary product of `skewed' selection functions. Iterations of the dependent binary products are also considered but in all cases are shown to be $T$-equivalent to the iteration of the simple products.
 \end{abstract}

%%%%%%%%%%%%%%%%%%%%%%%%%%%%%%%%%%%%%%%%
%%%%%%%%%%%%%%%%%%%%%%%%%%%%%%%%%%%%%%%%
\section{Introduction}
%%%%%%%%%%%%%%%%%%%%%%%%%%%%%%%%%%%%%%%%
%%%%%%%%%%%%%%%%%%%%%%%%%%%%%%%%%%%%%%%%

G\"odel's \cite{Goedel(58)} so-called dialectica interpretation reduces the consistency of Peano arithmetic to the consistency of the quantifier-free calculus of functionals $T$. In order to extend G\"odel's interpretation to full classical analysis $\PAomega + \CA$, Spector \cite{Spector(62)} made use of the fact that $\PAomega + \CA$ can be embedded, via the negative translation, into $\HAomega + \AC_\NN + \DNS$. Here $\PAomega$ and $\HAomega$ denote Peano and Heyting arithmetic, respectively, formulated in the language of finite types, and
\eqleft{\CA \; \colon \; \exists f^{\NN \to \BB} \forall n^\NN (f(n) \leftrightarrow A(n))}
is \emph{full comprehension},
\eqleft{\AC_\NN \; \colon \; \forall n^\NN \exists x^X A(n, x) \rightarrow \exists f \forall n A(n, f n)}
is \emph{countable choice}, and
\eqleft{\DNS \; \colon \; \forall n^\NN \neg \neg B(n) \rightarrow \neg \neg \forall n B(n),}
is the \emph{double negation shift}, with $A(n)$ and $A(n, x)$ standing for arbitrary formulas, and $B(n) \equiv \exists x \neg A(n, x)$. Since $\HAomega + \AC_\NN$, excluding the double negation shift, has a straightforward (modified) realizability interpretation \cite{Troelstra(73)}, as well as a dialectica interpretation \cite{Avigad(98),Goedel(58)}, the remaining challenge is to give a computational interpretation to $\DNS$.

A computational interpretation of $\DNS$ was first given by Spector \cite{Spector(62)}, via the dialectica interpretation. Spector devised a form of recursion on well-founded trees, nowadays known as \emph{Spector bar recursion}, and showed that the dialectica interpretation of $\DNS$ can be witnessed by such kind of recursion. A computational interpretation of $\DNS$ via realizability only came recently, first in \cite{Berardi(98)}, via a non-standard form of realizability, and then in \cite{BO(02A),BO(05)}, via Kreisel's modified realizability. The realizability interpretation of $\DNS$ makes use of a new form of bar recursion, termed \emph{modified bar recursion}. 

It has been shown in \cite{BO(05)} that Spector's bar recursion is definable in system $T$ extended with modified bar recursion, but not conversely, since Spector's bar recursion is S1-S9 computable in the model of total continuous functionals, but modified bar recursion is not. 

In the present paper we revisit these functional interpretations of classical analysis from the perspective of the newly developed theory of selection functions \cite{Escardo(2008),EO(2010A),EO(2010B),EO(2009),EO(2011A)}. \emph{Selection functionals} are functionals of type $(X \to R) \to X$, for arbitrary finite types $X, R$. We think of mappings $p \colon X \to R$ as generalised predicates, and of functionals $\varepsilon \colon (X \to R) \to X$ as witnessing, when possible, the ``non-emptiness" of any given such predicate. For instance, if $R = \BB$ is the set of booleans, Hilbert's $\varepsilon$-constant can be viewed as a selection function. Just as $\varepsilon$-terms in Hilbert's calculus can be used to define the existential quantifier, so can any selection function $\varepsilon \colon (X \to R) \to X$ be used to define a \emph{generalised quantifier} $\phi \colon (X \to R) \to R$ as
\[ \phi(p) \stackrel{R}{=} p(\varepsilon(p)). \]
Moreover, just like the usual quantifiers $\exists^X$ and $\forall^Y$ can be nested to produce a quantifier on the product space $X \times Y$, so can generalised quantifiers and selection functions. We prefer to think about the nesting of selection functions (and quantifiers) as a \emph{product operation}, since it transforms selection functions over spaces $X$ and $Y$ into a new selection function on the product space $X \times Y$ (cf. \cite{EO(2009)}). 

In this article we define two different iterations of the binary product of selection functions, one which we call \emph{implicitly controlled} and the other which we call \emph{explicitly controlled} (cf. also \cite{EO(2011A)}) and show that:
\begin{itemize}
	\item Modified bar recursion is $T$-equivalent to the implicitly controlled product of selection functions.
	\item Spector's bar recursion is $T$-equivalent to the explicitly controlled product of selection functions.
	\item The two different products can be used to interpret $\DNS$ directly via modified realizability and the dialectica interpretation, respectively.
	\item The implicitly controlled product of selection functions is strictly stronger than the explicitly controlled one.
	\item Apparently stronger iterations of the dependent products are in fact $T$-equivalent to the iterations of the simple products.
\end{itemize}

%%%%%%%%%%%%%%%%%%%%%%%%%%%%%%%%%%%%%%%%
%%%%%%%%%%%%%%%%%%%%%%%%%%%%%%%%%%%%%%%%
\section{Preliminaries}
%%%%%%%%%%%%%%%%%%%%%%%%%%%%%%%%%%%%%%%%
%%%%%%%%%%%%%%%%%%%%%%%%%%%%%%%%%%%%%%%%

Before we present our main results, let us first define the formal systems used, and give an introduction to our recent work on selection functions.

%%%%%%%%%%%%%%%%%%%%%%%%%%%%%%%%%%%%%%%%
\subsection{Heyting arithmetic $\HAomega$ and system $T$}
%%%%%%%%%%%%%%%%%%%%%%%%%%%%%%%%%%%%%%%%
\label{bar-ind}

% We work with a definitional extension of the finite types, allowing for infinite sequence of elements with different types. The main reason for this is to improve legibility of our constructions, as working with infinite sequence of a single type would obscure the shuffling of elements. This is formally defined as follows:

In this section we define the formal systems used to prove the inter-definability results. These include Heyting arithmetic in all finite types and extensions including bar induction and a continuity principle.

\begin{definition}[Finite types] The set of all \emph{finite types} $\Type$ are defined inductively as
\begin{itemize}
	\item $\BB$ (booleans) and $\NN$ (integers) are in $\Type$ \\[-3mm]
	\item If $X$ and $Y$ are in $\Type$ then $X \times Y$ (product) and $X \to Y$ (functions) are in $\Type$ \\[-3mm]
	\item If $X$ is in $\Type$ then $X^*$ (finite sequence) is in $\Type$.
\end{itemize}
We will also make informal use of the following type construction: Given a sequence of types $(X_i)_{i \in \NN}$ we also consider $\Pi_{i \in \NN} X_i$ and $\Pi_{i < n} X_i$ as types. The main purpose of this is to make the constructions more readable, since we can keep track of the positions which are being changed. A formal extension of system $T$ with such type construction has been considered by Tait \cite{tait(1965)}, hence we also hope that our presentation below will extend smoothly to a more general setting\footnote{If the reader prefers, however, she can assume that in a sequence of types $(X_i)_{i \leq \NN}$ all $X_i$ are equal $X$, replacing infinite sequence types $\Pi_{i \in \NN} X_i$ with $\NN \to X$, and finite sequence types $\Pi_{i < n} X_i$ with $X^*$.}, although in this paper we focus on the standard version of system $T$.
\end{definition}

% We also assume to have a unit type $\II$ so that $\Pi_{i < 0} X_i = \II$. 
We use $X, Y, Z$ for variables ranging over the elements of $\Type$. We often write $\Pi_i X_i$ for $\Pi_{i \in \NN} X_i$, and also $\Pi_{i \geq k} X_i$ for $\Pi_i X_{i+k}$.

Let $\HAomega$ be usual Heyting arithmetic in all finite types with a fully extensional treatment of equality, as in the system E-$\HAomega$ of \cite{Troelstra(73)}. Its quantifier-free fragment is the usual G\"odel's system $T$, also extended with sequence types. G\"odel's primitive recursion for each sequence of types $(X_i)_{i \in \NN} \in \Type$ is given by
\[
\begin{array}{lcl}
	\Rec f g 0 & \stackrel{X_0}{=} & g \\[2mm]
	\Rec f g (n+1) & \stackrel{X_{n+1}}{=} & f n (\Rec f g n)
\end{array}
\]
where $\Rec$ has finite type $\Pi_n (X_n \to X_{n+1}) \to X_0 \to \Pi_i X_i$. We also assume that we have a constant $\cZero^X$ of each finite type $X$, and the usual constructors and destructors such as $\pair{t^X, s^Y} \colon X \times Y$ and $\pi_i(\pair{s_0^{X_0}, s_1^{X_1}}) = s_i$, where $i = \{0,1\}$, for instance. For the newly introduced sequence types we have that if $t \colon \Pi_i X_i$ then $t i \colon X_i$; and if $t \colon X_i$ then $\lambda i . t \colon \Pi_i X_i$.  If $s \colon \Pi_{i < n} X_i$, we write $s_i \colon X_i$ for the $i$-th element of the sequence, for $i < n$. If $s \colon \Pi_{i < n} (X_i \times Y_i)$ is a sequence of pairs, we write $s^0 \colon \Pi_{i < n} X_i$ and $s^1  \colon \Pi_{i < n} Y_i$ for the projection of the sequence on the first and second coordinates, respectively. If $\alpha$ has type $\Pi_{i \in \NN} X_i$ we use the following abbreviations
\[
\begin{array}{lcl}
\alpha^n & \equiv & \lambda i . \alpha(i + n), \quad \mbox{(the $n$-left shift of $\alpha$, hence $\alpha^n \colon \Pi_i X_{i + n}$)} \\[2mm]
q^n(\alpha) & \equiv & q(\alpha^n), \quad \mbox{(so $q^n \colon \Pi_i X_i \to R$ if $q \colon \Pi_i X_{i + n} \to R$)} \\[2mm]
\alpha[k,n] & \equiv & \pair{\alpha(k), \ldots, \alpha(n)}, \quad \mbox{(finite segment from position $k$ to $n$)} \\[2mm]
\initSeg{\alpha}{n} & \equiv & \alpha[0, n-1], \quad \mbox{(initial segment of $\alpha$ of length $n$)} \\[2mm]
\initSegZ{\alpha}{n} & \equiv & \pair{\alpha(0), \ldots, \alpha(n-1), \cZero, \cZero, \ldots}, \; \mbox{(infinite extension of $\initSeg{\alpha}{n}$ with $\cZero$'s)} 
\end{array}
\]
where in the last case the type of $\cZero$ at the $i$-th coordinate is the same type of $\alpha(i)$.

We use $*$ for all forms of \emph{concatenation}. For instance, if $x$ has type $X_n$ and $s$ has type $\Pi_{i<n} X_i$ then $s * x$ is the concatenation of~$s$ with~$x$, which has type $\Pi_{i<n+1} X_i$. Similarly, if $x$ has type $X_0$ and $\alpha$ has type $\Pi_i X_{i+1}$ then $x * \alpha$ has type $\Pi_i X_i$. Given a functional $q \colon \Pi_i X_i \to R$ and a finite sequence $s \colon \Pi_{i <  n} X_i$ we write $q_s \colon \Pi_{i \geq n} X_i \to R$ for the function $\lambda \alpha . q(s * \alpha)$. When $s = \langle x \rangle$ we write $q_s$ as simply $q_x$. 

Given a finite sequence $s$ and an infinite sequence $\alpha$ let us write $s \overwrite \alpha$ for the ``overwriting" of $s$ on $\alpha$, i.e. $(s \overwrite \alpha)(i)$ equals $s_i$ if $i < |s|$ and equals $\alpha(i + |s|)$ otherwise.

% Finally, given two finite sequences $s$ and $t$ we write $s \preceq t$ to denote that $s$ is an initial segment of $t$.

In the following we shall assume that certain types are
\emph{discrete}. Semantically, in the model of total continuous
functionals, discreteness means that singletons are open or that all
points are isolated. Syntactically, the following grammar produces discrete types in that
model (along with compact types) \cite{Escardo(2008)}.

\begin{definition}[Discrete and compact types] Define the two subsets of $\Type$ inductively as follows:
\[
\begin{array}{lcl}
\compact & ::= & \BB \;|\; \compact \times \compact \;|\; \discrete \to \compact \\[2mm]
\discrete & ::= & \BB \;|\; \NN \;|\; \discrete \times \discrete \;|\; \discrete^* \;|\; \compact \to \discrete.
\end{array}
\]
\end{definition}

In this paper we work with a model independent notion of definability.
Formally, given a term $t$ in system~$T$, we view an equation $F(x) = t(F, x)$ as \emph{defining} or \emph{specifying} a functional~$F$.
We do not worry whether such an equation has a solution in any particular model of $\HAomega$, or whether it is unique, when it has a solution. 

\begin{definition} We say that a functional $G$ is $T$-definable from a functional $F$ (written $G \leq_T F$) over a theory $\mathcal{S}$ if there exists a term $s$ in system $T$ such that $s(F)$ satisfies the defining equation of $G$ provably in $\mathcal{S}$. We say that $F$ and $G$ are $T$-equivalent over $\mathcal{S}$, written $F =_T G$, if $G \geq_T F$ and $F \geq_T G$. 
\end{definition}

When stating in a theorem or proposition that $G$ is $T$-definable in $F$, we will explicitly write after the theorem/proposition number the theory $\mathcal{S}$ that is needed for the verification. In a few cases this theory will be an extension of $\HAomega$ with some the following three principles: \emph{Spector's condition}
\eqleft{\SPEC \; \colon \; \forall \omega^{\Pi_i X_i \to \NN} \forall \alpha^{\Pi_i X_i} \exists n (\omega(\initSegZ{\alpha}{n}) < n),}
the \emph{axiom of continuity}
\eqleft{\CONT \; \colon \; \forall q^{\Pi_i X_i \to R} \forall \alpha \exists n \forall \beta (\initSeg{\alpha}{n} \stackrel{\Pi_{i < n} X_i}{=} \initSeg{\beta}{n} \to q(\alpha) \stackrel{R}{=} q(\beta))}
with $R$ discrete, and the scheme of \emph{relativised bar induction} 
\eqleft{\BI \; \colon \; 
\left\{
\begin{array}{c}
S(\lempty) \\
\wedge \\
\forall \alpha \!\in\! S \, \exists n P(\initSeg{\alpha}{n}) \\
\wedge \\
\forall s \in S (\forall x [S(s * x)\to P(s * x)] \to P(s))
\end{array}
\right\} \to P(\lempty),
}
where $S(s)$ and $P(s)$ are arbitrary predicates in the language of $\HAomega$, and $\alpha\in S$ and $s\in S$ are shorthands for $\forall n S(\initSeg{\alpha}{n})$ and $S(s)$ respectively.  
We note that $\SPEC$ follows from $\CONT$, but it also holds in the model of strongly majorizable functionals \cite{Bezem(85)}.

\newcommand{\AAA}{\mathcal{A}}

%%%%%%%%%%%%%%%%%%%%%%%%%%%%%%%%%%%%%%%%
\subsection{Selection functions and generalised quantifiers}
%%%%%%%%%%%%%%%%%%%%%%%%%%%%%%%%%%%%%%%%
\label{binary:product}

In \cite{EO(2009),EO(2011A)} we have studied the properties of functionals having the type $(X \to R) \to R$, and called these \emph{generalised quantifiers}. When $R = \BB$ we have that $(X \to \BB) \to \BB$ is the type of the usual logical quantifiers $\forall, \exists$. We also showed that some generalised quantifiers $\phi \colon (X \to R) \to R$ are \emph{attainable}, in the sense that for some \emph{selection function} $\varepsilon \colon (X \to R) \to X$, we have
\[ \phi p = p(\varepsilon p) \]
for all (generalised) predicates $p$. In the case when $\phi$  is the usual existential quantifier, for instance, $\varepsilon$ corresponds to Hilbert's epsilon term. 
Since the types $(X \to R) \to R$ and $(X \to R) \to X$ will be used quite often, we abbreviate them as $K_R X$ and $J_R X$, respectively. Moreover, when $R$ is fixed, we often simply write $K X$ and $J X$, omitting the subscript $R$. In \cite{EO(2009)} we also defined products of quantifiers and selection functions.

\begin{definition}[Product of selection functions and quantifiers] \label{main-simple} Given generalised quantifiers $\phi \colon K X$ and $\psi \colon K Y$, define the product quantifier $(\phi \ttimes \psi) \colon K (X \times Y)$ as
\[ (\phi \ttimes \psi)(p^{X \times Y \to R}) \stackrel{R}{\fdefin} \phi(\lambda x^X . \psi (\lambda y^Y . p(x, y))). \]
Also, given selection functions $\varepsilon \colon J X$ and $\delta \colon J Y$, define the product selection function $(\varepsilon \ttimes \delta) \colon J(X \times Y)$ as
\[ (\varepsilon \ttimes \delta)(p^{X \times Y \to R}) \stackrel{X \times Y}{\fdefin} (a, b(a)) \]
where
\eqleft{
\begin{array}{lcl}
a & \stackrel{X}{\fdefin} & \varepsilon(\lambda x^X. p(x, b(x))) \\[2mm]
b(x^X) & \stackrel{Y}{\fdefin} & \delta(\lambda y^Y . p(x, y)).
\end{array}
}
\end{definition}

One of the results we obtained is that the product of attainable quantifiers is also attainable. This follows from the fact that the product of quantifiers corresponds to the product of selection functions, as made precise in the following lemma.

\begin{lemma}[\cite{EO(2009)}, Lemma 3.1.2] \label{basic} Let $R$ be fixed. Given a selection function $\varepsilon : J X$, define a quantifier $\selEmb{\varepsilon} \colon K X$ as
\eqleft{\selEmb{\varepsilon} p \fdefin p(\varepsilon p).}
Then for $\varepsilon \colon J X$ and $\delta \colon J Y$ we have $\selEmb{\varepsilon \ttimes \delta} = \selEmb \varepsilon \ttimes \selEmb{\delta}.$
\end{lemma}

Given a finite sequence of selection functions or quantifiers, the two binary products defined above can be iterated so as to give rise to finite products of selection functions and quantifiers. We have shown that such a construction also appears in game theory (backward induction), algorithms (backtracking), and proof theory (interpretation of the infinite pigeon-hole principle) -- see \cite{EO(2009)} for details.

In the following (Sections \ref{conditional} and \ref{implicit}) we
will describe two possible ways of iterating the binary product of
selection function an infinite, or unbounded, number of times.
% But before that, let us recall two principles which we use later.

%%%%%%%%%%%%%%%%%%%%%%%%%%%%%%%%%%%%%%%%
%%%%%%%%%%%%%%%%%%%%%%%%%%%%%%%%%%%%%%%%
\section{Explicitly Controlled Product}
%%%%%%%%%%%%%%%%%%%%%%%%%%%%%%%%%%%%%%%%
%%%%%%%%%%%%%%%%%%%%%%%%%%%%%%%%%%%%%%%%
\label{conditional}

The finite product of selection functions of Definition \ref{main-simple} can be infinitely iterated in two ways. The first, which we define in this section is via an \emph{explicitly controlled} iteration, which we will show to correspond to Spector's bar recursion. In the following section we also define an \emph{implicitly controlled} iteration, which we will show to correspond to modified bar recursion.

\begin{definition}[$\eps$] \label{eps-def} Let $\varepsilon \colon \Pi_k J X_k$ be a sequence of selection functions. Define their \emph{explicitly controlled infinite product} as
\[
\eps_n^\length(\varepsilon)(q) \stackrel{\Pi_i X_{i + n}}{=} 
\left\{
\begin{array}{ll}
\cZero & {\rm if} \; \length(q(\cZero)) < n \\[2mm]
(\varepsilon_n \ttimes \eps_{n+1}^\length(\varepsilon))(q) \quad & {\rm otherwise},
\end{array}
\right.
\]
where $q \colon \Pi_i X_{i+n} \to R$ and $\length \colon R \to \NN$. We call $\length$ the \emph{length function} since it controls the length of the recursive path. Unfolding the definition of $\ttimes$ we can write the defining equation of $\eps$ as
\begin{equation} \label{eps-eq-def} \tag{\eps}
\eps_n^\length(\varepsilon)(q) \stackrel{\Pi_i X_{i + n}}{=} 
\left\{
\begin{array}{ll}
\cZero & {\rm if} \; \length(q(\cZero)) < n \\[2mm]
c * \eps_{n+1}^\length(\varepsilon)(q_c) \quad & {\rm otherwise},
\end{array}
\right.
\end{equation}
where $c = \varepsilon_n(\lambda x . \selEmb{\eps_{n+1}^\length(\varepsilon)}(q_x))$.
\end{definition}

The next lemma (essentially Lemma 1 of \cite{Spector(62)}) states one of the most crucial properties of this product of selection functions.

\begin{lemma}[$\HAomega + (\ref{eps-eq-def})$] \label{spector-main-lemma} Let $\alpha \fdefin \eps_n^{\length}(\varepsilon)(q)$. Then, for all $i \colon \NN$
\[ \alpha = \initSeg{\alpha}{i} * \eps_{n + i}^\length(\varepsilon)(q_{\initSeg{\alpha}{i}}). \]
\end{lemma}
\begin{proof} By induction on $i$. If $i = 0$ this follows by the definition of $\alpha$. Assume this holds for $i$, we wish to show it also holds for $i+1$. Consider two cases. \\[2mm]
If $\length(q_{\initSeg{\alpha}{i}}(\cZero)) = \length(q(\initSegZ{\alpha}{i})) < n + i$ then
\begin{itemize}
	\item[$(i)$] $\eps_{n + i}^\length(\varepsilon)(q_{\initSeg{\alpha}{i}}) = \cZero$
\end{itemize}
and hence
\begin{itemize}
	\item[$(ii)$] $\alpha \stackrel{\textup{(IH)}}{=} \initSeg{\alpha}{i} * \eps_{n+i}^\length(\varepsilon)(q_{\initSeg{\alpha}{i}}) \stackrel{(i)}{=} \initSegZ{\alpha}{i} = \initSegZ{\alpha}{i+1}$.
\end{itemize}
Therefore,
\begin{itemize}
	\item[$(iii)$] $\length(q(\initSegZ{\alpha}{i+1})) \stackrel{(ii)}{=} \length(q(\initSegZ{\alpha}{i})) < n + i < n + i + 1$.
\end{itemize}
Hence, by $(iii)$ we have
\begin{itemize}
	\item[$(iv)$] $\eps_{n + i + 1}^\length(\varepsilon)(q_{\initSeg{\alpha}{i+1}}) = \cZero$.
\end{itemize}
So
\eqleft{\alpha \stackrel{(ii)}{=} \initSegZ{\alpha}{i+1} \stackrel{(iv)}{=} \initSeg{\alpha}{i+1} * \eps_{n + i + 1}^\length(\varepsilon)(q_{\initSeg{\alpha}{i+1}}).}
On the other hand, if $\length(q_{\initSeg{\alpha}{i}}(\cZero)) = \length(q(\initSegZ{\alpha}{i})) \geq n + i$, then
\eqleft{\alpha \stackrel{\textup{(IH)}}{=} \initSeg{\alpha}{i} * \eps_{n+i}^\length(\varepsilon)(q_{\initSeg{\alpha}{i}}) = \initSeg{\alpha}{i} * c * \eps_{n+i+1}^\length(\varepsilon)(q_{\initSeg{\alpha}{i} * c}),}
so that $\alpha(i) = c$. Hence $\alpha = \initSeg{\alpha}{i+1} * \eps_{n+i+1}^\length(\varepsilon)(q_{\initSeg{\alpha}{i + 1}})$. 
\end{proof}

An immediate consequence of the lemma above is that it allows us to calculate the $i$-th element of the infinite sequence $\eps_n^\length(\varepsilon)(q)$ (see also Theorem \ref{main-spec} for another important consequence). 

\begin{corollary}[$\HAomega + (\ref{eps-eq-def})$] \label{unwinding-cps} For all $n$ and $i$
\[ \eps_n^\length(\varepsilon)(q)(i) \stackrel{X_{n + i}}{=} 
\left\{
\begin{array}{ll}
\cZero & {\rm if} \; \length(q_t(\cZero)) < n + i \\[2mm]
\varepsilon_{n + i}(\lambda x . \selEmb{\eps_{n + i + 1}^\length(\varepsilon)}(q_{t*x})) \;\; & {\rm otherwise},
\end{array}
\right.
\]
where $t = \initSeg{\eps_n^\length(\varepsilon)(q)}{i}$.
\end{corollary}
\begin{proof} Let $\alpha = \eps_n^\length(\varepsilon)(q)$ so that $t = \initSeg{\eps_n^\length(\varepsilon)(q)}{i} = \initSeg{\alpha}{i}$. By Lemma \ref{spector-main-lemma} we have that $\alpha(i) =  \eps_{n+i}^\length(\varepsilon)(q_t)(0)$. Hence, by ($\ref{eps-eq-def}$) we have the desired result. \end{proof}

The fact that $\eps$ exists in the model of total continuous functionals, and is in fact uniquely characterized by its defining equation, can be seen as follows. First, note that the $\eps_n^\length(\varepsilon)(q)$ is an infinite sequence, say $\alpha \colon \Pi_i X_{i + n}$. Intuitively, at each recursive call the functional $q$ gets information about one more element of its input sequence. Assuming continuity we will have that $\length \circ q \colon \Pi_i X_{i + n} \to \NN$ will eventually always return a fixed value, no matter what the rest of the input sequence is. This means that as $n$ increases we will eventually have $\length(q(\cZero)) < n$. It is perhaps surprising that such a functional also exists in the model of strongly majorizable functionals~\cite{Bezem(85)}, which contains discontinuous functionals! Following the construction of Bezem \cite{Bezem(85)} one can prove this directly, but this result will also follow from our result that $\eps$ is $T$-definable from Spector's bar recursion (Section \ref{sec-spector}).

We also define the corresponding explicitly controlled product of \emph{quantifiers} as follows:

\begin{definition}[$\epq$] \label{def-epq} Let $\phi \colon \Pi_k K X_k$ be a sequence of quantifiers. Their \emph{explicitly controlled infinite product} is defined as
\[ \epq_n^\length(\phi)(q) \stackrel{R}{=} 
\left\{
\begin{array}{ll}
q(\cZero) & {\rm if} \; \length(q(\cZero)) < n \\[2mm]
(\phi_n \ttimes \epq_{n+1}^\length(\phi))(q) & {\rm otherwise},
\end{array}
\right.
\]
where $q \colon \Pi_i X_{i + n} \to R$ and $\length \colon R \to \NN$. Unfolding the definition of the binary product of quantifiers we have
\begin{equation} \label{def-epq-equation} \tag{\epq}
\epq_n^\length(\phi)(q) \stackrel{R}{=} 
\left\{
\begin{array}{ll}
q(\cZero) & {\rm if} \; \length(q(\cZero)) < n \\[2mm]
\phi_n(\lambda x^{X_n} . \epq_{n+1}^\length(\phi)(q_x)) & {\rm otherwise}.
\end{array}
\right.
\end{equation}
\end{definition}

Howard (proof attributed to Kreisel) shows in Lemma 3C of \cite{Howard(1968)} that assuming Spector's bar recursion one can prove Spector's stopping condition $\SPEC$. It is easy to see that the form of bar recursion used by Howard is also an instance of $\epq$ and hence we obtain:

\begin{lemma} \label{howard-kreisel} $\HAomega + (\ref{def-epq-equation}) \vdash \SPEC$.
\end{lemma}

We now show that $\epq$ and $\eps$ are $T$-equivalent. That $\epq$ is $T$-definable in $\eps$ has been recently shown in \cite{Oliva(2012A)}. Hence, it remains to show that $\epq$ defines $\eps$. The proof makes use of the fact that each selection function $\varepsilon$ defines a quantifier, as $\phi(p) = p(\varepsilon(p))$ (cf. Lemma \ref{basic}). In order to define $\eps$ for the types $(X_i, R)$ we shall use $\epq$ for the types $(X_i, R')$ with $R' = \Pi_i X_i$.

\begin{lemma}[$\HAomega + \BI + \SPEC$] \label{epq-eps-lemma} Let $R' = \Pi_i X_i$. Given $\varepsilon_i \colon J_R X_i$ and $q \colon \Pi_i X_i \to R$ define
\begin{equation} \label{epq-eps-lemma-eq}
\phi_i^{\varepsilon, q}(p^{X_i \to R'}) \stackrel{R'}{\fdefin} p(\varepsilon_i(\lambda x^{X_i} . q(p(x)))).
\end{equation}
Defined also the sequence of functions $f^n = \lambda \alpha . \cZr{\Pi_{i < n} X_i} * \alpha$. Then (with $q \colon \Pi_{i \geq n} X_i \to R$ so $q^n \colon \Pi_i X_i \to R$)
\[ \epq_{n+1}^{\length \circ q^n}(\phi^{\varepsilon, q^n})((f^n)_x) = (\cZero * x^{X_n}) \overwrite \epq_{n + 1}^{\length \circ (q_x)^{n+1}}(\phi^{\varepsilon, (q_x)^{n+1}})(f^{n+1}). \]
\end{lemma}

\begin{proof} By bar induction $\BI$ and the axiom $\SPEC$. We take $S(s) = {\sf true}$ and
\[ \underbrace{\epq_{n+|s|+1}^{\length \circ q^n}(\phi^{\varepsilon, q^n})((f^n)_{x * s}) = (\cZero * x) \overwrite \epq_{n+|s|+1}^{\length \circ (q_x)^{n+1}}(\phi^{\varepsilon, (q_x)^{n+1}})((f^{n+1})_s)}_{P(s)} \]
where $s \colon \Pi_{n < i \leq n + |s|} X_i$. \\[1mm]
($i$) $\forall \alpha \exists j \, P(\initSeg{\alpha}{j})$. By $\SPEC$, for any $\alpha \colon \Pi_{i > n} X_i$ there is a point $j$ such that
\[
\begin{array}{lcl}
(\length \circ q^n)(\cZr{\Pi_{i < n} X_i} * x * \initSeg{\alpha}{j} * \cZero) & = & \length(q(x * \initSeg{\alpha}{j} * \cZero)) \\[1mm]
	& < & n + j + 1.
\end{array}
\]
For such $j$ and $s = \initSeg{\alpha}{j}$ it is easy to see that $P(s)$ holds as both sides of $P(s)$ are equal to $\cZero * x * s * \cZero$. \\[1mm]
($ii$) $\forall s (\forall y P(s * y) \to P(s))$. Let $s$ be such that $\forall y P(s * y)$; we show $P(s)$. We can assume that
\[ (\length \circ q^n \circ (f^n)_{x * s})(\cZero) = (\length \circ (q_x)^{n+1} \circ (f^{n+1})_s)(\cZero) \geq n + |s| + 1, \]
as otherwise the proof can be carried out as in case $(i)$ above. Hence, we calculate
\[
\begin{array}{l}
	\epq_{n+|s|+1}^{\length \circ q^n}(\phi^{\varepsilon, q^n})((f^n)_{x * s}) \\[1.5mm]
	\quad \;\;\; \stackrel{(\ref{def-epq-equation})}{=} \phi^{\varepsilon, q^n}_{n+|s|+1}(\lambda y . \epq_{n+|s|+2}^{\length \circ q^n}(\phi^{\varepsilon, q^n})((f^n)_{x * s * y}) \\[1.5mm]
	\quad \quad \,\stackrel{(\ref{epq-eps-lemma-eq})}{=} \epq_{n+|s|+2}^{\length \circ q^n}(\phi^{\varepsilon, q^n})((f^n)_{x * s * c}) \\[1.5mm]
	\quad \quad \,\stackrel{\textup{(IH)}}{=} (\cZero * x) \overwrite \epq_{n+|s|+2}^{\length \circ (q_x)^{n+1}}(\phi^{\varepsilon, (q_x)^{n+1}})((f^{n+1})_{s * c}) \\[1.5mm]
	\quad \quad\; \stackrel{(*)}{=} (\cZero * x) \overwrite \epq_{n+|s|+2}^{\length \circ (q_x)^{n+1}}(\phi^{\varepsilon, (q_x)^{n+1}})((f^{n+1})_{s * \tilde{c}}) \\[1.5mm]
	\quad \quad\; \stackrel{(\ref{epq-eps-lemma-eq})}{=} (\cZero * x) \overwrite \phi_{n+|s|+1}^{\varepsilon, (q_x)^{n+1}}(\lambda y . \epq_{n+|s|+2}^{\length \circ (q_x)^{n+1}}(\phi^{\varepsilon, (q_x)^{n+1}})((f^{n+1})_{s * y})) \\[1.5mm]
	\quad \quad\, \stackrel{(\ref{def-epq-equation})}{=} (\cZero * x) \overwrite \epq_{n+|s|+1}^{\length \circ (q_x)^{n+1}}(\phi^{\varepsilon, (q_x)^{n+1}})((f^{n+1})_s)
\end{array}
\]
where
\begin{itemize}
	\item[] $c = \varepsilon_{n+|s|+1}(\lambda y . q^n(\epq_{n+|s|+2}^{\length \circ q^n}(\phi^{\varepsilon, q^n})((f^n)_{x * s * y})))$
	\item[] $\tilde{c} = \varepsilon_{n+|s|+1}(\lambda y . (q_x)^{n+1}(\epq_{n+|s|+2}^{\length \circ (q_x)^{n+1}}(\phi^{\varepsilon, (q_x)^{n+1}})((f^{n+1})_{s * y})))$
\end{itemize}
so that $(*) \; c = \tilde{c}$ follows directly from the induction hypothesis $\forall y P(s * y)$.
\end{proof}

We are now ready to show that $\eps$ is $T$-definable from $\epq$. The proof presented here is essentially the same as Spector's proof that his restricted form of bar recursion $\SBR$ follows from the general form $\BR$ (cf. Section \ref{sec-spector}). 

% The reason we need to make use of Lemma \ref{epq-eps-lemma} (and hence bar induction) is due to updating of the functional $q$ in the unwindings of $\epq$ and $\eps$. 

\begin{theorem}[$\HAomega + \BI$] \label{cps-sbr} $\epq \geq_T \eps$.
\end{theorem}
\begin{proof} Let $\phi^{\varepsilon, q}_n$ and $f^n$ be as defined in Lemma \ref{epq-eps-lemma}. We claim that $\eps$ can be defined from $\epq$ as
\begin{itemize}
\item[$(i)$] $\eps_n^\length(\varepsilon)(q) \stackrel{\Pi_i X_{i+n}}{\fdefin} (\epq_n^{\length \circ q^n}(\phi^{\varepsilon, q^n})(f^n))^n$.
\end{itemize}
We consider two cases. \\[1mm]
If $\length(q(\cZero^{\Pi_{i \geq n} X_i})) < n$ then we also have $\length(q^n(\cZero^{\Pi_i X_i})) < n$. Therefore
\eqleft{\eps_n^\length(\varepsilon)(q)
		\,\stackrel{(i)}{=}\, (\epq_n^{\length \circ q^n}(\phi^{\varepsilon, q^n})(f^n))^n
		\,\stackrel{(\ref{def-epq-equation})}{=}\, (f^n(\cZero))^n
		\,=\, \cZr{\Pi_i X_{i + n}}.
}
On the other hand, if $\length(q(\cZero^{\Pi_{i \geq n} X_i})) \geq n$ then $\length(q^n(\cZero^{\Pi_i X_i})) \geq n$ and hence:
\eqleft{
\begin{array}{lcl}
	\eps_n^\length(\varepsilon)(q)
		& \stackrel{(i)}{=} & (\epq_n^{\length \circ q^n}(\phi^{\varepsilon, q^n})(f^n))^n \\[1mm]
		& \stackrel{(\ref{def-epq-equation})}{=} & (\phi_n^{\varepsilon, q^n}(\lambda x^{X_n} . \epq_{n+1}^{\length \circ q^n}(\phi^{\varepsilon, q^n})((f^n)_x)))^n \\[1.5mm]
		& \stackrel{(\ref{epq-eps-lemma-eq})}{=} & (\epq_{n+1}^{\length \circ q^n}(\phi^{\varepsilon, q^n})((f^n)_c))^n \\[1.5mm]
		& \stackrel{\textup{L}\ref{epq-eps-lemma}}{=} & ((\cZero * c) \overwrite \epq_{n + 1}^{\length \circ (q_c)^{n+1}}(\phi^{\varepsilon, (q_c)^{n+1}})(f^{n+1}))^n \\[1.5mm]
		& = & c * (\epq_{n+1}^{\length \circ (q_c)^{n+1}}(\phi^{\varepsilon, (q_c)^{n+1}})(f^{n+1}))^{n+1} \\[1mm]
		& \stackrel{(i)}{=} & c * \eps_{n+1}^\length(\varepsilon)(q_c)
\end{array}
}
where
\eqleft{
\begin{array}{lcl}
c 
	& \fdefin & \varepsilon_n(\lambda x^{X_n} . q^n(\epq_{n+1}^{\length \circ q^n}(\phi^{\varepsilon, q^n})((f^n)_x))) \\[1mm]
	& \stackrel{\textup{L}\ref{epq-eps-lemma}}{=} & \varepsilon_n(\lambda x^{X_n} . q^n((\cZero * x) \overwrite \epq_{n+1}^{\length \circ (q_x)^{n+1}}(\phi^{\varepsilon, (q_x)^{n+1}})(f^{n+1}))) \\[1.5mm]
	& = & \varepsilon_n(\lambda x^{X_n} . (q_x)^{n+1}(\epq_{n+1}^{\length \circ (q_x)^{n+1}}(\phi^{\varepsilon, (q_x)^{n+1}})(f^{n+1}))) \\[1.5mm]
	& = & \varepsilon_n(\lambda x^{X_n} . q_x(\epq_{n+1}^{\length \circ (q_x)^{n+1}}(\phi^{\varepsilon, (q_x)^{n+1}})(f^{n+1}))^{n+1}) \\[1.5mm]
	& \stackrel{(i)}{=} & \varepsilon_n(\lambda x^{X_n} . q_x(\eps_{n+1}^{\length}(\varepsilon)(q_x))).
\end{array}
}
\end{proof}

%%%%%%%%%%%%%%%%%%%%%%%%%%%%%%%%%%%%%%%%
\subsection{Dialectica interpretation of classical analysis}
%%%%%%%%%%%%%%%%%%%%%%%%%%%%%%%%%%%%%%%%
\label{dialectica}

In order to find witnesses for the dialectica interpretation of $\DNS$, and hence full classical analysis, Spector arrived at the following system of equations
\begin{equation} \label{spector-eq}
\begin{array}{lcl}
	n & \stackrel{\NN}{=} & \omega \alpha, \\[2mm]
	\alpha(n) & \stackrel{X}{=} & \varepsilon_n(p), \\[2mm]
	p(\alpha(n)) & \stackrel{R}{=} & q \alpha,
\end{array}
\end{equation}
where $\varepsilon_n \colon J_R X$ and $q \colon (\NN \to X) \to R$ and $\omega \colon (\NN \to X) \to \NN$ are given and $n \colon \NN$ and $\alpha \colon \NN \to X$ and $p \colon X \to R$ are the unknowns.
We now show how $\eps$ can be used to solve Spector's equations. We first solve a slightly different set of equations, and as a corollary we obtain a solution to Spector's original one.

\begin{theorem}[$\HAomega + (\ref{eps-eq-def})$] \label{main-spec} Let $q \colon \Pi_i X_i \to R$ and $\length \colon R \to \NN$ and $\varepsilon \colon \Pi_i J_R X_i$ be given. Define
\eqleft{
\begin{array}{lcl}
\alpha & \fdefin & \eps_0^{\length}(\varepsilon)(q) \\[2mm]
p_n(x) & \fdefin & \selEmb{\eps^\length_{n+1}(\varepsilon)}(q_{\initSeg{\alpha}{n} * x}).
\end{array}
}
For $n \leq \length(q(\alpha))$ we have
\[
\begin{array}{lcl}
	\alpha(n) & \stackrel{X_n}{=} & \varepsilon_n(p_n) \\[2mm]
	p_n(\alpha(n)) & \stackrel{R}{=} & q \alpha.
\end{array}
\]
\end{theorem}
\begin{proof} This is essentially Spector's proof (cf. lemma 11.5 of \cite{Kohlenbach(2008)}). Assume $n \leq \length(q(\alpha))$. We first argue that $n \leq \length(q(\initSegZ{\alpha}{n}))$. Otherwise, assuming $n > \length(q(\initSegZ{\alpha}{n})) = \length(q_{\initSeg{\alpha}{n}}(\cZero))$ we would have, by Lemma \ref{spector-main-lemma}, that $\alpha = \initSegZ{\alpha}{n}$. And hence, by extensionality, $n > \length(q_{\initSeg{\alpha}{n}}(\cZero)) = \length(q(\alpha)) \geq n$, which is a contradiction. \\[2mm]
%
% Then, it 
Hence, assuming $n \leq \length(q(\alpha))$ we have $n \leq \length(q(\initSegZ{\alpha}{n}))$ and hence
\eqleft{\alpha(n) \,\stackrel{\textup{C}\ref{unwinding-cps}}{=}\, \varepsilon_n(\lambda x . \selEmb{\eps_{n + 1}^\length(\varepsilon)}(q_{\initSeg{\alpha}{n} * x})) \,=\, \varepsilon_n(p_n).}
For the second equality, we have
\eqleft{
\begin{array}{lcl}
	p_n(\alpha(n)) & = & \selEmb{\eps_{n+1}^\length(\varepsilon)}(q_{\initSeg{\alpha}{n+1}}) \\[2mm]
		        	& = & q_{\initSeg{\alpha}{n+1}}(\eps_{n+1}^\length(\varepsilon)(q_{\initSeg{\alpha}{n+1}})) \\[2mm]
		        	& = & q(\initSeg{\alpha}{n+1}  * \eps_{n+1}^\length(\varepsilon)(q_{\initSeg{\alpha}{n+1}})) \\[1mm]
			& \stackrel{\textup{L}\ref{spector-main-lemma}}{=} & q(\alpha).
\end{array}
}
\end{proof}

\begin{corollary} \label{spector-solution} For any given $q \colon X^\NN \to R$ and $\omega \colon X^\NN \to \NN$ and sequence of selection functions $\varepsilon_n \colon J_R X$ (of common type $J_R X$) there are $\alpha \colon \NN \to X$ and $p \colon X \to R$ satisfying the system of equations (\ref{spector-eq}).
\end{corollary}
\begin{proof} Let $R' = R \times \NN$, and let $\pi_0 \colon R \times \NN \to R$ and $\pi_1 \colon R \times \NN \to \NN$ denote the first and second projections. Define
\eqleft{
\begin{array}{lcl}
	q'(\alpha) & \stackrel{R'}{=} & \pair{q(\alpha), \omega(\alpha)} \\[2mm] 
	\varepsilon_n'(p^{X \to R'}) & \stackrel{X}{=} & \varepsilon_n(\lambda x^{X} . \pi_0(p(x))).
\end{array}
}
so $q' \colon (\NN \to X) \to R'$ and $\varepsilon_n' \colon J_{R'} X$. Let
\eqleft{
\begin{array}{lcl}
\alpha & \stackrel{X^\NN}{=} & \eps_0^{\pi_1}(\varepsilon')(q') \\[2mm]
p_n'(x^X) & \stackrel{R'}{=} & \selEmb{\eps_{n+1}^{\pi_1}(\varepsilon')}(q_{\initSeg{\alpha}{n} * x}').
\end{array}
}
Assume $n \leq \omega(\alpha) = \pi_1(q'(\alpha))$. By Theorem \ref{main-spec} we have 
\[
\begin{array}{lcl}
	\alpha(n) & \stackrel{X}{=} & \varepsilon_n'(p_n') \\[2mm]
	p_n'(\alpha(n)) & \stackrel{R'}{=} & q' \alpha.
\end{array}
\]
Finally, let $n = \omega(\alpha)$ and $p(x) = \pi_0(p_n'(x))$. Then it is easy to check that $\alpha$ and $p$ satisfy the desired equation, e.g. $\alpha(n) = \varepsilon_n'(p_n') = \varepsilon_n(\pi_0 \circ p_n') = \varepsilon_n(p)$.
\end{proof}

%%%%%%%%%%%%%%%%%%%%%%%%%%%%%%%%%%%%%%%%
\subsection{Dependent variants of $\eps$ and $\epq$}
%%%%%%%%%%%%%%%%%%%%%%%%%%%%%%%%%%%%%%%%
\label{EPS-eps}

In the dialectica interpretation of $\DNS$ given above (Section \ref{dialectica}), the selection functions $\varepsilon_n$ do not depend on the history of choices already made. Thus, it was sufficient to use an iteration of the \emph{simple} product of selection functions. Nevertheless, Spector bar recursion and modified bar recursion are normally formulated in the most general form, where selection functions at point $n$ have access to the values $X_i$ for $i < n$. 

In the same vein, in previous papers \cite{EO(2010A),EO(2009)} we have considered generalisations of the product of selection functions, where a selection function (or a quantifier) at stage $n$ can have access to the previously computed values. We called these the \emph{dependent} product of selection functions and quantifiers. 

\begin{definition}[Dependent product of selection functions and quantifiers] \label{main-dependent} Given a quantifier $\phi \colon K X$ and a family of quantifiers $\psi \colon X \to K Y$, define the dependent product quantifier $(\phi \ttimesD \psi) \colon K (X \times Y)$ as
\[ (\phi \ttimesD \psi)(p^{X \times Y \to R}) \stackrel{R}{\fdefin} \phi(\lambda x^X . \psi (x, \lambda y^Y . p(x, y))). \]
Also, given a selection function $\varepsilon \colon J X$ and a family of selection functions $\delta \colon X \to J Y$, define the  dependent product selection function $(\varepsilon \ttimesD \delta) \colon J(X \times Y)$ as
\[ (\varepsilon \ttimesD \delta)(p^{X \times Y \to R}) \stackrel{X \times Y}{\fdefin} (a, b(a)) \]
where
\eqleft{
\begin{array}{lcl}
a & \stackrel{X}{\fdefin} & \varepsilon(\lambda x^X. p(x, b(x))) \\[2mm]
b(x) & \stackrel{Y}{\fdefin} & \delta(x, \lambda y^Y . p(x, y)).
\end{array}
}
\end{definition}

As done for the simple product of selection functions and quantifiers, we can also iterate the dependent products as follows:

\begin{definition}[$\EPQ$ and $\EPS$] \label{definition-EPS} Given a family of quantifiers \[ \phi \colon \Pi_k (\Pi_{i < k} X_i \to K X_k), \] define their \emph{dependent explicitly controlled product} (denoted $\EPQ$) as
\[
\EPQ_s^\length(\phi)(q) \stackrel{R}{=} 
\left\{
\begin{array}{ll}
q(\cZero) & {\rm if} \;\length(q(\cZero)) < |s| \\[2mm]
(\phi_s \ttimesD (\lambda x^{X_{|s|}} . \EPQ_{s * x}^\length(\phi)))(q) & {\rm otherwise}.
\end{array}
\right.
\]
Unpacking the definition of the binary dependent product $\ttimesD$ this is equivalent to
\begin{equation} \label{EPQ-def} \tag{\EPQ}
\EPQ_s^\length(\phi)(q) \stackrel{R}{=} 
\left\{
\begin{array}{ll}
q(\cZero) & {\rm if} \;\length(q(\cZero)) < |s| \\[2mm]
\phi_s(\lambda x^{X_{|s|}} . \EPQ_{s * x}^\length(\phi)(q_x)) & {\rm otherwise}.
\end{array}
\right.
\end{equation}
Moreover, given a family of selection functions \[ \varepsilon \colon \Pi_k (\Pi_{i < k} X_i \to J X_k), \] define their \emph{dependent explicitly controlled product} (denoted $\EPS$) as
\[
\EPS_s^\length(\varepsilon)(q) \stackrel{\Pi_i X_{|s|+i}}{=} 
\left\{
\begin{array}{ll}
\cZero & {\rm if} \;\length(q(\cZero)) < |s| \\[2mm]
(\varepsilon_s \ttimesD (\lambda x^{X_{|s|}} . \EPS_{s * x}^\length(\varepsilon)))(q) & {\rm otherwise}.
\end{array}
\right.
\]
Similarly, unfolding the definition of $\ttimesD$ the defining equation for $\EPS$ is equivalent to
\begin{equation} \label{EPS-def} \tag{\EPS}
\EPS_s^\length(\varepsilon)(q) \stackrel{\Pi_i X_{|s|+i}}{=} 
\left\{
\begin{array}{ll}
\cZero & {\rm if} \;\length(q(\cZero)) < |s| \\[2mm]
c * \EPS_{s * c}^\length(\varepsilon)(q_c) & {\rm otherwise}
\end{array}
\right.
\end{equation}
where $c = \varepsilon_s(\lambda x . \selEmb{\EPS_{s * x}^\length(\varepsilon)}(q_x))$. 
\end{definition}

%Similar to Lemma \ref{unwinding-cps} we can also show that
%%
%\begin{equation} \label{EPS-eq-def}
%\EPS_s^\length(\varepsilon)(q)(i) =
%\left\{
%\begin{array}{ll}
%\cZero & {\rm if} \; \length(q_t(\cZero)) < |s| + i \\[2mm]
%\varepsilon_{s*t}(\lambda x . \selEmb{\EPS_{s * t * x}^\length(\varepsilon)}(q_{t*x})) & {\rm otherwise},
%\end{array}
%\right.
%\end{equation}
%%
%where $t = \initSeg{\EPS_s^\length(\varepsilon)(q)}{i}$.

Clearly $\eps$ is $T$-definable from $\EPS$. We now show that in fact $\eps$ and $\EPS$ are $T$-equivalent. In Theorem \ref{EPS-from-eps} we will make use of the following construction. Given $\alpha \colon \Pi_{i \geq n}(\Pi_{k < i} X_k \to X_i)$ and $s \colon \Pi_{k<n} X_k$ define $\alpha^s \colon \Pi_{i \geq n} X_i$ by course-of-values as
\begin{equation} \label{s-construction}
\alpha^s(i) \stackrel{X_{n+i}}{=} \alpha(i)(s * \initSeg{\alpha^s}{i}).
\end{equation}
Clearly, given a finite sequence $t \colon \Pi_{i \in [n, m]}(\Pi_{k < i} X_k \to X_i)$ we can perform the same construction to obtain a $t^s \colon \Pi_{i \in [n, m]} X_i$. 

\begin{lemma}[$\HAomega$] \label{s-constr-lemma} $(d * \alpha)^s = d(s) * (\alpha)^{s * d(s)}$, where $d \colon \Pi_{k < n} X_k \to X_n$.
\end{lemma}
\begin{proof} Straightforward.
\end{proof}

Finally, given $q \colon \Pi_{i \geq n} X_i \to R$ define $q^s \colon \Pi_{i \geq n}(\Pi_{k < i} X_k \to X_i) \to R$ as
\eqleft{
\begin{array}{lcl}
q^s(\alpha) & \stackrel{R}{\fdefin} & q(\alpha^s). 
\end{array}
}

\begin{theorem}[$\HAomega$] \label{EPS-from-eps} $\eps \geq_T \EPS$.
\end{theorem}
\begin{proof} To define $\EPS$ of type $(X_k, R)$ we use $\eps$ of type $(\Pi_{i < k} X_i \to X_k, R)$. Given selection functions $\varepsilon_s \colon J X_{|s|}$ define $\tilde{\varepsilon}_k \colon J(\Pi_{i<k} X_i \to X_k)$ as
\begin{itemize}
\item[$(i)$] $\tilde{\varepsilon}_k(P^{(\Pi_{i<k} X_i \to X_k) \to R}) \stackrel{\Pi_{i<k} X_i \to X_k}{\fdefin} \lambda s^{\Pi_{i<k} X_i} . \varepsilon_s(\lambda y^{X_k}. P(\lambda t. y))$.
\end{itemize}
Note that the infinite (simple) product of the selection functions $\tilde{\varepsilon}_k$ has type
\eqleft{\eps_n^\length(\tilde{\varepsilon}) \; \colon \; J(\Pi_{i \geq n}(\Pi_{k < i} X_k \to X_i))}
where $\length \colon R \to \NN$. We claim that $\EPS$ can be defined from $\eps$ as
\begin{itemize}
\item[$(ii)$] $\EPS_s^\length(\varepsilon)(q^{\Pi_i X_{i + |s|} \to R}) \stackrel{\Pi_i X_{i + |s|}}{\fdefin} (\eps_{|s|}^\length(\tilde{\varepsilon})(q^s))^s$,
\end{itemize}
where $s \colon \Pi_{k<|s|} X_k$. Let us show that $\EPS$ as defined above satisfies the defining equation ($\ref{EPS-def}$). 
Consider two cases: \\[1mm]
If $\length(q(\cZero)) < |s|$ then $\length(q^s(\cZero)) = \length(q(\cZero)) < |s|$. Hence, by the definition of $(\cdot)^s$
\eqleft{\EPS_s^\length(\varepsilon)(q) \stackrel{(ii)}{=} (\eps_{|s|}^\length(\tilde{\varepsilon})(q^s))^s \stackrel{(\ref{eps-eq-def})}{=} (\cZero)^s = \cZero.}
On the other hand, if $\length(q(\cZero)) \geq |s|$ then $\length(q^s(\cZero)) = \length(q(\cZero)) \geq |s|$. Hence
\eqleft{
\begin{array}{lcl}
\EPS_s^\length(\varepsilon)(q)
	& \stackrel{(ii)}{=} & (\eps_{|s|}^\length(\tilde{\varepsilon})(q^s))^s \\[1mm]
	& \stackrel{(\ref{eps-eq-def})}{=} & (d * \eps_{|s |+1}^\length(\tilde{\varepsilon})((q^s)_d))^s \\[1mm]
	& \stackrel{\textup{L}\ref{s-constr-lemma}}{=} & d(s) * (\eps_{|s|+1}^\length(\tilde{\varepsilon})((q^s)_d))^{s * d(s)} \\[1mm]
	& \stackrel{\textup{L}\ref{s-constr-lemma}}{=} & d(s) * (\eps_{|s|+1}^\length(\tilde{\varepsilon})((q_{d(s)})^{s * d(s)}))^{s * d(s)} \\[1mm]
	& \stackrel{(*)}{=} & c * (\eps_{|s*c|}^\length(\tilde{\varepsilon})((q_c)^{s * c}))^{s * c} \\[1mm]
	& \stackrel{(ii)}{=} & c * \EPS_{s * c}^\length(\varepsilon)(q_c),
\end{array}
}
where $d = \tilde{\varepsilon}_{|s|}(\lambda f . (q^s)_f(\eps_{|s|+1}^\length(\tilde{\varepsilon})((q^s)_f)))$ and $c = \varepsilon_s(\lambda x . q_x(\EPS_{s * x}^\length(\varepsilon)(q_x))$ so that $(*)$
\eqleft{
\begin{array}{lcl}
d(s)	& = & \tilde{\varepsilon}_{|s|}(\lambda f . (q^s)_f(\eps_{|s|+1}^\length(\tilde{\varepsilon})((q^s)_f)))(s) \\[1mm]
	& \stackrel{(i)}{=} & \varepsilon_s(\lambda x . (q^s)_{\lambda t . x}(\eps_{|s| + 1}^\length(\tilde{\varepsilon})((q^s)_{\lambda t . x}))) \\[1mm]
	& \stackrel{\textup{L}\ref{s-constr-lemma}}{=} & \varepsilon_s(\lambda x . (q_x)^{s * x}(\eps_{|s| + 1}^\length(\tilde{\varepsilon})((q_x)^{s * x}))) \\[2mm]
	& = & \varepsilon_s(\lambda x . q_x(\eps_{|s * x|}^\length(\tilde{\varepsilon})((q_x)^{s * x}))^{s * x}) \\[1mm]
	& \stackrel{(ii)}{=} & \varepsilon_s(\lambda x . q_x(\EPS_{s * x}^\length(\varepsilon)(q_x)) \\[2mm]
	& = & c.
\end{array}
}
\end{proof}

\begin{remark} Note that a similar construction does not work in the case of quantifiers, since there is no $\lambda$-term (in the pure simply typed $\lambda$-calculus) of type $(X \to K Y) \to K(X \to Y)$, for arbitrary $X$ and $Y$. Nevertheless, it will follow from our results that the explicitly controlled iteration of the simple product of quantifiers $\epq$ is $T$-equivalent to the explicitly controlled iteration of the dependent product of quantifiers $\EPQ$ (cf. summary of results in Figure \ref{table}). 
\end{remark}

%%%%%%%%%%%%%%%%%%%%%%%%%%%%%%%%%%%%%%%%
\subsection{Relation to Spector's bar recursion}
%%%%%%%%%%%%%%%%%%%%%%%%%%%%%%%%%%%%%%%%
\label{sec-spector}

As we have shown in Theorem \ref{main-spec}, which is essentially Spector's solution, the explicitly controlled product of selection functions $\eps$ can also be used to give a computational interpretation of classical analysis. When presenting his solution in \cite{Spector(62)}, Spector first formulates a general ``construction by bar recursion" as
\begin{equation} \label{BR-def} \tag{\BR}
\BR_s^\omega(\phi)(q) \stackrel{R}{=} 
\left\{
\begin{array}{ll}
q_s(\cZero) & {\rm if} \; \omega_s(\cZero) < |s| \\[2mm]
\phi_s(\lambda x^{X_{|s|}} . \BR_{s*x}^\omega(\phi)(q)) \quad & {\rm otherwise},
\end{array}
\right.
\end{equation}
where $\phi_s \colon K_R X_{|s|}$, $q \colon \Pi_i X_i \to R$ and $\omega \colon \Pi_i X_i \to \NN$. This is usually referred to as \emph{Spector's bar
  recursion}, but we argue that this is misleading. We show that $\BR$ is
closely related to the product of \emph{quantifiers} $\EPQ$, whereas
the special case of this used by Spector is equivalent to the
(dependent) product of \emph{selection functions} $\EPS$, which we
have shown to be equivalent to $\eps$ (Section \ref{EPS-eps}).

\begin{remark} In fact, Spector's definition seems slightly more general than $\BR$ as defined here, since in Spector's definition $q$ might also depend on the length of the sequence $s$. As we show in Lemma \ref{spec-search}, however, it is possible to reconstruct $|s|$ from the sequence $s * \cZero$ if $s$ is the point where Spector's condition first happens.
\end{remark}

\begin{theorem}[$\HAomega + \BI$] \label{BR-EPQ-def} $\BR \geq_T \EPQ$.
\end{theorem}
\begin{proof} In order to define $\EPQ$ of type $(X_i, R)$ we use $\BR$ of the same type $(X_i, R$). $\BR$ and $\EPQ$ have very similar definitions, except that in $\BR$ the stopping condition is given directly on the current sequence $s * \cZero$, whereas in $\EPQ$ a ``length" function $\length \colon R \to \NN$ is used so that the stopping condition involves the composition $\length \circ q$. Hence, in order to define $\EPQ$ from $\BR$ it is essentially enough to take $\omega = \length \circ q$, taking care of the fact that the types of $q$ in $\EPQ$ and $\BR$ are slightly different as $q$ in $\EPQ$ takes a ``shorter" input sequence starting at point $|s|$. We show that $\EPQ$ defined as
\begin{itemize}
\item[$(i)$] $\EPQ_s^{\length}(\phi)(q) \fdefin \BR_s^{\length \circ q^{|s|}}(\phi)(q^{|s|})$
\end{itemize}
satisfies the equation $(\ref{EPQ-def})$. Consider two cases. \\[1mm]
If $(\length \circ q^{|s|})_s(\cZero) = \length(q(\cZero)) < |s|$ then
\eqleft{
\EPQ_s^{\length}(\phi)(q) \stackrel{(i)}{=} \BR_s^{\length \circ q^{|s|}}(\phi)(q^{|s|}) \stackrel{(\ref{BR-def})}{=} (q^{|s|})_s(\cZero) = q(\cZero).
}
On the other hand, if $(\length \circ q^{|s|})_s(\cZero) = \length(q(\cZero)) \geq |s|$ then
\eqleft{
\begin{array}{lcl}
	\EPQ_s^\length(\phi)(q)
		& \stackrel{(i)}{=} & \BR_s^{\length \circ q^{|s|}}(\phi)(q^{|s|}) \\[1.5mm]
		& \stackrel{(\ref{BR-def})}{=} & \phi_s(\lambda x^{X_{|s|}} . \BR_{s*x}^{\length \circ q^{|s|}}(\phi)(q^{|s|})) \\[1mm]
		& \stackrel{(*)}{=} & \phi_s(\lambda x^{X_{|s|}} . \BR_{s * x}^{\length \circ (q_x)^{|s*x|}}(\phi)((q_x)^{|s*x|}) \\[2mm]
		& = & \phi_s(\lambda x^{X_{|s|}} . \EPQ_{s * x}^\length(\phi)(q_x))
\end{array}
}
where
\begin{itemize}
	\item[$(*)$] $\BR_{s*x}^{\length \circ q^{|s|}}(\phi)(q^{|s|}) =  \BR_{s * x}^{\length \circ (q_x)^{|s * x|}}(\phi)((q_x)^{|s*x|})$
\end{itemize}
can, as in Lemma~\ref{epq-eps-lemma}, be proved by bar induction $\BI$ and axiom $\SPEC$, since
\eqleft{q^{|s|}(s * x * \alpha) = (q_x)^{|s * x|}(s * x * \alpha).}
Finally, recall that $\HAomega + (\ref{BR-def}) \vdash \SPEC$ (Lemma 3C of \cite{Howard(1968)}).
\end{proof}

Spector, however, explicitly says that only a \emph{restricted form} of $\BR$ is used for the dialectica interpretation of (the negative translation of) countable choice. It is this restricted form that we shall from now on call \emph{Spector's bar recursion}:

\begin{definition}[Spector's bar recursion] Spector's bar recursion \cite{Spector(62)} is the
recursion schema
\[
\SBR_s^\omega(\varepsilon) \stackrel{\Pi_i X_i}{=} 
s
\overwrite
\left\{
\begin{array}{ll}
\cZero & {\rm if} \; \omega_s(\cZero) < |s| \\[2mm]
\SBR_{s * c}^\omega(\varepsilon) & {\rm otherwise},
\end{array}
\right.
\]
where $c \stackrel{X_{|s|}}{=} \varepsilon_{s}(\lambda x^{X_{|s|}} . \SBR_{s * x}^\omega(\varepsilon))$, and where $\varepsilon_s \colon J_{\Pi_i X_i} X_{|s|}$ and $\omega \colon \Pi_i X_i \to \NN$.
\end{definition}

We now show that Spector's bar recursion is $T$-definable from the explicitly controlled product of selection functions $\EPS$. It will follow from other results that they are in fact $T$-equivalent (see Figure \ref{table}).

\begin{theorem}[$\HAomega$] \label{main-eps-sbr} $\EPS \geq_T \SBR$.
\end{theorem}
\begin{proof} To define $\SBR$ of type $(X_i)$ we use $\EPS$ of type $(X_i, (\Pi_i X_i) \times \NN)$. $\EPS$ and $\SBR$ have very similar definitions, except that $\EPS$ has an extra argument $q \colon \Pi_{i \geq |s|} X_i \to R$. We can obtain $\SBR$ from $\EPS$ by simply taking $q(\alpha)$ to be the identity function plus the stopping value $\omega(\alpha)$. So, the length function $\length \colon R \to \NN$ can be taken to be the second projection. The details are as follows: Let $R = (\Pi_i X_i) \times \NN$. Given $\omega \colon \Pi_i X_i \to \NN$ and $\varepsilon_s \colon J_{\Pi_i X_i} X_{|s|}$, define
\eqleft{
\begin{array}{lcl}
\length(r^R) & \stackrel{\NN}{\fdefin} & \pi_1(r) \\[1mm]
q^\omega(\alpha^{\Pi_i X_i}) & \stackrel{R}{\fdefin} & \pair{\alpha, \omega(\alpha)} \\[1mm]
\tilde \varepsilon_s (p^{X_{|s|} \to R}) & \stackrel{X_{|s|}}{=} & \varepsilon_s(\pi_0 \circ p).
\end{array}
}
Define
\begin{itemize}
\item[$(i)$] $\SBR_s^\omega(\varepsilon) \stackrel{\Pi_i X_i}{=} s * \EPS_s^\length(\tilde \varepsilon)((q^\omega)_s)$.
\end{itemize}
If $\omega_s(\cZero) < |s|$ then $\length((q^\omega)_s(\cZero)) = \omega_s(\cZero) < |s|$ and hence
\eqleft{
\begin{array}{lcl}
	\SBR_s^\omega(\varepsilon)
		& \stackrel{(i)}{=} & s * \EPS_s^\length(\tilde \varepsilon)((q^\omega)_s) \\[0.5mm]
		& \stackrel{(\ref{EPS-def})}{=} & s * \cZero \\[1.5mm]
		& = & s \overwrite \cZero.
\end{array}
}
On the other hand, if $\omega_s(\cZero) \geq |s|$ then also $\length((q^\omega)_s(\cZero)) \geq |s|$ and we have 
\eqleft{
\begin{array}{lcl}
	\SBR_s^\omega(\varepsilon)
		& \stackrel{(i)}{=} & s * \EPS_s^\length(\tilde \varepsilon)((q^\omega)_s) \\[0.5mm]
		& \stackrel{(\ref{EPS-def})}{=} & s * c * \EPS_{s * c}^\length(\tilde \varepsilon)((q^\omega)_{s * c}) \\[1mm]
		& \stackrel{(*)}{=} & s * d * \EPS_{s * d}^\length(\tilde \varepsilon)((q^\omega)_{s * d}) \\[1mm]
		& \stackrel{(i)}{=} & \SBR_{s*d}^\omega(\varepsilon)
\end{array}
}
where $c =  \tilde \varepsilon_s(\lambda x . (q^\omega)_{s * x}(\EPS_{s * x}^\length(\tilde \varepsilon)((q^\omega)_{s * x})))$ and $d = \varepsilon_s(\lambda x . \SBR_{s * x}^\omega(\varepsilon))$ so that $(*)$
\eqleft{
\begin{array}{lcl}
c 	& = & \tilde \varepsilon_s(\lambda x . (q^\omega)_{s * x}(\EPS_{s * x}^\length(\tilde \varepsilon)((q^\omega)_{s * x}))) \\[1mm]
	& \stackrel{\textup{def} \; q^\omega,\, \tilde \varepsilon_s}{=} & \varepsilon_s(\lambda x . s * x * \EPS_{s * x}^\length(\tilde \varepsilon)((q^\omega)_{s * x})) \\[1mm]
	& \stackrel{(i)}{=} & \varepsilon_s(\lambda x . \SBR_{s * x}^\omega(\varepsilon)) \;=\; d.
\end{array}
}
\end{proof}

%%%%%%%%%%%%%%%%%%%%%%%%%%%%%%%%%%%%%%%%
%%%%%%%%%%%%%%%%%%%%%%%%%%%%%%%%%%%%%%%%
\section{Implicitly Controlled Product}
%%%%%%%%%%%%%%%%%%%%%%%%%%%%%%%%%%%%%%%%
%%%%%%%%%%%%%%%%%%%%%%%%%%%%%%%%%%%%%%%%
\label{implicit}

We have seen in Section \ref{conditional} above that the explicitly
controlled iterated product of selection functions is sufficient to
witness the dialectica interpretation of the double negation shift (and hence, classical countable
choice). In this section we show that when interpreting this same
principle via modified realizability, one seems to need an
\emph{unrestricted} or, as we we shall call it, \emph{implicitly
  controlled} infinite product of selection functions.

\begin{definition}[$\ips$] The \emph{implicitly controlled product} of a sequence of selection functions $\varepsilon \colon \Pi_k J X_k$ is defined as
\[ \ips_n(\varepsilon) \stackrel{J (\Pi_i X_{i + n})}{=} \varepsilon_n \ttimes \ips_{n+1}(\varepsilon). \]
Unfolding the definition of $\ttimes$, this is the same as
\begin{equation} \label{ips-def-eq} \tag{\ips}
\ips_n(\varepsilon)(q) \stackrel{\Pi_i X_{i + n}}{=} \underbrace{\varepsilon_n(\lambda x . q_x(\ips_{n+1}(\varepsilon)(q_x)))}_{c} * \,\ips_{n+1}(\varepsilon)(q_c).
\end{equation}
\end{definition}

We call the above infinite product \emph{implicitly controlled} because under the assumption of continuity for $q \colon \Pi_i X_{i + n} \to R$, for discrete $R$, the bar recursive calls % will 
eventually terminate. 

\begin{remark} As shown in section 5.6. of \cite{EO(2009)}, an implicitly controlled product of quantifiers $\ipq$
\[
\ipq_n(\phi) = \phi_n \ttimes \ipq_{n+1}(\phi)
\]
does not exist. It is enough to consider the case when $R = X_i = \NN$. Let $\phi_n(p) = 1 + p(0)$ and $q$ be any function. Assuming the equation above, it follows by induction that for all $n$
\[
\ipq_0(\phi)(q) = n + \ipq_n(\phi)(q_{0^n}),
\]
where $0^n = \langle 0, 0, \ldots, 0 \rangle$, with $n$ zeros; which implies $\ipq_0(\phi)(q) \geq n$, for all $n$.
\end{remark}

%Unwinding the definition of the binary product we obtain a lemma similar to Corollary \ref{unwinding-cps} for $\eps$, except in this case the lemma is much easier to prove as we do not have to consider the case distinctions.
%
%\begin{lemma}[$\HAomega + (\ref{ips-def-eq})$] \label{unwinding-ps} For all $n$ and $i$
%%
%\[
%\ips_n(\varepsilon)(q)(i) \stackrel{X_{n + i}}{=}  \varepsilon_{n + i}(\lambda x^{X_{n + i}} . \selEmb{\ips_{n + i + 1}(\varepsilon)}(q_{s *x})) 
%\]
%%
%where $s = \initSeg{\ips_n(\varepsilon)(q)}{i}$.
%\end{lemma}
%%
%\begin{proof} By unfolding the definition of the binary product of selection functions using course-of-values induction. \end{proof}

%%%%%%%%%%%%%%%%%%%%%%%%%%%%%%%%%%%%%%%%
\subsection{Realizability interpretation of classical analysis}
%%%%%%%%%%%%%%%%%%%%%%%%%%%%%%%%%%%%%%%%
\label{realizability}

We now describe how $\ips$ can be used to interpret the double negation shift (and hence countable choice) via modified realizability. As discussed in the introduction, a computational interpretation of full classical analysis can be reduced to an interpretation of the double negation shift $\DNS$. Given that the formula $A(n)$ (in $\DNS$) can be assumed to be of the form $\exists x \neg B(n, x)$, $\DNS$ is equivalent to
\[
\forall n ((A(n) \to \perp) \to A(n)) \to (\forall n A(n) \to \perp) \to \forall n A(n).
\]
That is because, for $A(n) \equiv \exists x \neg B(n, x)$, we have both $\perp \, \to A(n)$ and $\perp \, \to \forall n A(n)$ in \emph{minimal logic}. Moreover, since the negative translation brings us into minimal logic, falsity $\perp$ can be replaced by an arbitrary $\Sigma^0_1$-formula~$R$. This is known as the (refined) $A$-translation \cite{Berger(95)}, and is useful to analyse proofs of $\Pi^0_2$ theorems in analysis. Recall that we are using the abbreviation 
\eqleft{J_R A = (A \to R) \to A.}
The resulting principle we obtain is what we shall call the $J$-shift
\eqleft{\NUS \quad \colon \quad \forall n J_R A(n) \to J_R \forall n A(n).}
$\DNS$ is then the particular case of the $K$-${\sf shift}$
\eqleft{\KSHIFT \quad \colon \quad \forall n K_R A(n) \to K_R \forall n A(n),}
when $R = \perp$; considering the other type construction
\eqleft{K_R A = (A \to R) \to R.} 
One advantage of moving to the $\NUS$ is that $A(n)$ now can be taken to be an arbitrary formula, not necessarily of the form $\exists x \neg B(n, x)$. Hence the principle $\NUS$ is more general than $\DNS$. We analyse the logical strength of the principle $\NUS$ in more detail in \cite{EO(2010B)}, where a proof translation based on the construction $J_R A$ is also defined.
Our proof of the following theorem is very similar to that of \cite[Theorem 3]{BO(02A)}. 
We assume continuity and relativised bar induction as formulated in Section \ref{bar-ind}.

\begin{theorem}[$\HAomega + \BI + \CONT$] \label{ips-jshift} $\ips_0$ modified realizes $\NUS$.
\end{theorem}
\begin{proof} Given a term $t$ and a formula $A$ we write ``$t \mr A$" for ``$t$ modified realizes $A$" (see \cite{Troelstra(73)} for definition). Assume that
\eqleft{
\begin{array}{lcl}
	\varepsilon_n & \mr & (A(n) \to R) \to A(n), \\[2mm]
	q & \mr & \forall n A(n) \to R.
\end{array} }
Let
\eqleft{
\begin{array}{lcl}
P(s) & \equiv & s * \ips_{|s|}(\varepsilon)(q_s) \,\mr\,  \forall n A(n) \\[2mm]
S(s) & \equiv & \forall n \!<\! |s| \, (s_n \mr A(n)).
\end{array}
}
We show $P(\langle \, \rangle)$ by bar induction relativised to the predicate $S$. Let us write $\alpha \in S$ as an abbreviation for $\forall n (\initSeg{\alpha}{n} \in S)$. The first assumption of $\BI$ (i.e. $S(\pair{\,})$) is vacuously true. We now prove the other two assumptions. \\[2mm]
($i$) $\forall \alpha \!\in\! S \, \exists k \, P(\initSeg{\alpha}{k})$. Given $\alpha \in S$ let $k$ be a point of continuity of $q$ on $\alpha$. Let $r := q \alpha$. By $\CONT$ we have $q(\initSeg{\alpha}{k} * \beta) = r$, for all $\beta$. By the assumptions on $\alpha$ and $q$ we have that $r \mr R$. We must show that for all $n$
\[ (\initSeg{\alpha}{k} * \ips_k(\varepsilon)(\lambda \beta . r))(n) \,\mr\,  A(n), \]
If $n < k$ this follows directly from the assumption $\alpha \in S$. In case $n \geq k$ we must show $\varepsilon_n(\lambda x . c) \, \mr \, A(n)$, which follows from the assumptions on $\varepsilon_n$ and $r$. \\[2mm]
($ii$) $\forall s \!\in\! S (\forall x [S(s * x) \to P(s * x)] \to P(s))$. Assume $\forall x [S(s * x) \to P(s * x)]$ with $s \!\in\! S$. We must prove $P(s)$, i.e.
\[ s * \ips_{|s|}(\varepsilon)(q_s) \,\mr\, \forall n A(n). \]
Unfolding the definition of $\ips_{|s|}$ (cf. $(\ref{ips-def-eq})$) this is equivalent to 
\[ \underbrace{s * c * \ips_{|s| + 1}(\varepsilon)(q_{s * c}) \,\mr\, \forall n A(n)}_{P(s * c)}. \]
where $c = \varepsilon_{|s|}(\lambda x. q_{s * x}(\ips_{|s| + 1}(\varepsilon)(q_{s * x})))$. Since $s \in S$, by the bar induction hypothesis it is enough to show that $c \mr A(|s|)$, i.e.
\[ \varepsilon_{|s|}(\lambda x. q_{s * x}(\ips_{|s| + 1}(\varepsilon)(q_{s * x}))) \mr A(|s|) \]
so that also $s * c \in S$. By the assumption on $\varepsilon_{|s|}$, the above follows from
\[ \lambda x. q_{s * x}(\ips_{|s| + 1}(\varepsilon)(q_{s * x})) \mr A(|s|) \to R. \]
Finally, by the assumption on $q$ the above is a consequence of
\[ \underbrace{s * x * \ips_{|s| + 1}(\varepsilon)(q_{s * x}) \,\mr\, \forall n A(n)}_{P(s * x)}, \]
for $x \mr A(|s|)$, which follows from the (bar) induction hypothesis. \end{proof}

%%%%%%%%%%%%%%%%%%%%%%%%%%%%%%%%%%%%%%%%
\subsection{Dependent variant of $\ips$}
%%%%%%%%%%%%%%%%%%%%%%%%%%%%%%%%%%%%%%%%
\label{IPS-ips}

Consider also the implicitly controlled dependent product of selection functions.

\begin{definition}[$\IPS$] \label{definition-IPS} Let $\varepsilon \colon \Pi_k (\Pi_{i < k} X_i \to J X_k)$. Define the \emph{dependent implicitly controlled product of selection functions} (denoted $\IPS$) as
\begin{equation} \label{IPS-def} \tag{\IPS}
\IPS_s(\varepsilon) \stackrel{J (\Pi_i X_{|s|+i})}{=} \varepsilon_s \ttimesD (\lambda x^{X_{|s|}} . \IPS_{s * x}(\varepsilon)).
\end{equation}
%
%As in Lemma \ref{unwinding-ps} for $\ips$, we can also equivalently define $\IPS$ as
%%
%\[
%\IPS_s(\varepsilon)(q)(i) \stackrel{X_{i + |s|}}{=}  \varepsilon_{s*t}(\lambda x^{X_{i + |s|}} . \selEmb{\IPS_{s * t * x}(\varepsilon)}(q_{t *x})), 
%\]
%%
%where $t = \initSeg{\IPS_s(\varepsilon)(q)}{i}$.
\end{definition}

Again (similar to Section \ref{EPS-eps}), it is clear that $\IPS$ is a generalisation of $\ips$. We now show that the proof that $\IPS$ is $T$-definable from $\ips$ can be easily adapted to show that also $\ips$ is $T$-equivalent to its dependent variant $\IPS$. In fact, in the case of $\ips$ and $\IPS$ the proof is slightly simpler since we do not have to worry about the stopping condition and the length function $\length$. 

\begin{theorem}[$\HAomega$] \label{IPS-from-ips} $\ips \geq_T \IPS$.
\end{theorem}
\begin{proof} Let $\tilde{\varepsilon}_k$ be as defined in Theorem \ref{EPS-from-eps}. Note that the infinite (simple) product of selection functions applied to $\tilde{\varepsilon}$ has type
\eqleft{\ips_n(\tilde{\varepsilon}) \;\; \colon \;\; J(\Pi_{i \geq n}(\Pi_{k < i} X_k \to X_i)).}
We claim that $\IPS$ can then be defined from $\ips$ as
\begin{itemize}
	\item[$(i)$] $\IPS_s(\varepsilon)(q^{\Pi_j X_{j + |s|} \to R}) \stackrel{\Pi_j X_{j + |s|}}{\fdefin} (\ips_{|s|}(\tilde{\varepsilon})(q^s))^s$
\end{itemize}
where $s \colon \Pi_{k<|s|} X_k$ and $(\cdot)^s$ is as defined in (\ref{s-construction}). We have
\[
\begin{array}{lcl}
\IPS_s(\varepsilon)(q) 
	& \stackrel{(i)}{=} & (\ips_{|s|}(\tilde{\varepsilon})(q^{s}))^{s} \\[1mm]
	& \stackrel{(\ref{ips-def-eq})}{=} & (d * \ips_{|s|+1}(\tilde{\varepsilon})((q^{s})_d))^{s} \\[1mm]
	& \stackrel{\textup{L}\ref{s-constr-lemma}}{=} & d(s) * (\ips_{|s|+1}(\tilde{\varepsilon})((q^{s})_d))^{s * d(s)} \\[2mm]
	& \stackrel{\textup{L}\ref{s-constr-lemma}}{=} & d(s) * (\ips_{|s|+1}(\tilde{\varepsilon})((q_{d(s)})^{s * d(s)}))^{s * d(s)} \\[2mm]
	& \stackrel{(*)}{=} & c * (\ips_{|s * c|}(\tilde{\varepsilon})((q_c)^{s * c}))^{s * c} \\[2mm]
	& \stackrel{(i)}{=} & c * \IPS_{s * c}(\varepsilon)(q_c)
\end{array}
\]
where, as in Theorem \ref{EPS-from-eps}, we can show that $(*) \; d(s) = c$ for
\eqleft{
\begin{array}{lcl}
d & = & \tilde{\varepsilon}_{|s|}(\lambda f . (q^s)_f(\ips_{|s|+1}(\tilde{\varepsilon})((q^s)_f))) \\[2mm]
c & = & \varepsilon_s(\lambda x . (q_x)(\IPS_{s * x}(\varepsilon)(q_x)).
\end{array}
}
\end{proof}

%%%%%%%%%%%%%%%%%%%%%%%%%%%%%%%%%%%%%%%%
\subsection{Relation to modified bar recursion}
%%%%%%%%%%%%%%%%%%%%%%%%%%%%%%%%%%%%%%%%
\label{sec-mbr}

The proof that $\ips$ interprets full classical analysis, via modified realizability, is very similar to the one given in \cite{BO(02A),BO(05)} that \emph{modified bar recursion} $\MBR$ interprets full classical analysis. In this section we show how $\MBR$ corresponds directly to the infinite iteration of a different form of binary product of selection functions. We also show (cf. Section \ref{mbr-def-ips}) that this different product when iterated leads to a form of bar recursion ($\MBR$) which is nevertheless $T$-equivalent to $\IPS$.

\begin{definition} \label{skewed-prod} Given a function $\varepsilon \in (X \to R) \to X \times Y$ and a selection function $\delta \in J Y$ define a selection function $\varepsilon \mttimes  \delta \in J(X \times Y)$ as
\[
(\varepsilon \mttimes \delta)(p) \stackrel{X \times Y}{\fdefin} \varepsilon(\lambda x. p(x, b(x)))
\]
where $b(x) \stackrel{Y}{\fdefin} \delta(\lambda y . p(x, y))$. We shall also consider a dependent version $\mttimesD$ of the product where $\delta \colon X \to J Y$ and $b(x) \fdefin \delta(x, \lambda y . p(x, y))$.
\end{definition}

The above construction shows how a mapping of type $(X \to R) \to X \times Y$ can be extended to a selection function on the product space, given a selection function on $Y$. We shall use this with $X = X_n$ and $Y = \Pi_i X_{i + n + 1}$, so that we obtain a selection function in $J(\Pi_i X_{i+n})$.

\begin{definition}[$\mbr$] Let $\varepsilon_n \colon (X_n \to R) \to \Pi_i X_{i + n}$ and $\varepsilon = (\varepsilon_n)_{n \in \NN}$. Define the \emph{iterated skewed product} $\isp$ as
\[ \mbr_n(\varepsilon) \stackrel{J (\Pi_i X_{i + n})}{=} \varepsilon_n \mttimes \mbr_{n+1}(\varepsilon).\]
Unfolding the definition of $\mttimes$ we have
\begin{equation} \label{mbr-def-eq} \tag{\mbr}
\mbr_n(\varepsilon)(q) \stackrel{\Pi_i X_{i + n}}{=} \varepsilon_n(\lambda x . q_x( \mbr_{n+1}(\varepsilon)(q_x) ).
\end{equation}
Define also the \emph{dependent iterated skewed product} $\MBR$
\[
\MBR_s(\varepsilon) \stackrel{J (\Pi_i X_{i + |s|})}{=} \varepsilon_s \mttimesD (\lambda x . \MBR_{s * x}(\varepsilon)),
\]
where in this case $\varepsilon_s \colon (X_{|s|} \to R) \to \Pi_i X_{i + |s|}$. Again, unfolding the definition of $\mttimes$ we have
\begin{equation} \label{MBR-def-eq-unfold} \tag{\MBR}
\MBR_s(\varepsilon)(q) \stackrel{\Pi_i X_{i + |s|}}{=} \varepsilon_s(\lambda x . q_x( \MBR_{s * x}(\varepsilon)(q_x)).
\end{equation}
We name this $\isp$ and $\MBR$ because we will show this is essentially \emph{modified bar recursion} as defined in \cite{BO(02A),BO(05)}.
\end{definition}

We think of $\varepsilon$ as a sequence of \emph{skewed selection functions}. The idea is that sometimes a witness for $X_k$ is automatically a witness for all types $X_i$ for $i \geq k$. In such cases, a selection function $\varepsilon_n \colon (X_n \to R) \to X_n$ gives rise to a skewed selection function $\varepsilon_n \colon (X_n \to R) \to \Pi_{i \geq n} X_i$, so that the more intricate product of selection functions (Definition \ref{main-simple}) can be replaced by the simpler product given in Definition \ref{skewed-prod}. 

Similarly to $\IPS$ and $\EPS$ (Sections \ref{EPS-eps} and \ref{IPS-ips}), we now show that the \emph{simple} iterated skewed product $\mbr$ is $T$-equivalent to its \emph{dependent} variant $\MBR$. Given a sequence of types $X_i$ let us define the new sequence
\[
Y_j \equiv \BB \times (\Pi_{i < j} X_i \to \Pi_k X_{k + j}).
\]
The intuition for the construction below is the same as the one used to show that $\eps$ $T$-defines $\EPS$ (Theorem \ref{EPS-from-eps}), except that here we need an extra boolean flag as the whole result of the skewed selection function will be returned on the first position of the output. The flag is used so that functions querying such sequences can know which are proper values and which are dummy values. Let us first define the construction $(\cdot)^{[s]} \colon \Pi_i Y_{i + |s|} \to \Pi_i X_{i + |s|}$ that given $\alpha \colon \Pi_i Y_{i + |s|}$ is defined as
\eqleft{
\alpha^{[s]}(i) \stackrel{X_{i + |s|}}{\fdefin} 
\begin{cases}
g(s * \initSeg{\alpha^{[s]}}{n})(i-n) & \text{if $\exists n \leq i (\alpha(n) = \pair{\cTrue, g})$} \\[2mm]
\cZero & \text{if $\forall n \leq i (\alpha(n) = \pair{\cFalse, \ldots})$}
  \end{cases}
}
where in the first case $n$ is the greatest $n \leq i$ such that $\alpha(n)$ is of the form $\pair{\cTrue, g}$. Finally, for $q \colon \Pi_i X_{i + |s|} \to R$ we define
\eqleft{q^{[s]}(\alpha) \stackrel{R}{\fdefin} q(\alpha^{[s]})}
so $q^{[s]} \colon \Pi_i Y_{i + |s|} \to R$; and for $x \colon X_j$ define $\hat x \colon Y_j$ as
\eqleft{\hat x = \pair{\cTrue, \lambda s^{\Pi_{i < j} X_i} . \pair{x^{X_j}, \cZero^{X_{j+1}}, \cZero^{X_{j+2}}, \ldots}}.}

\begin{lemma}[$\HAomega$] \label{isp-ISP-lemma} $(q^{[s]})_{\hat x} =  (q_x)^{[s * x]}$. 
\end{lemma}

\begin{proof} By course-of-values induction on $i$ we have $(\hat x * \alpha)^{[s]}(i) = (x * \alpha^{[s * x]})(i)$. Hence $(\hat x * \alpha)^{[s]} = x * \alpha^{[s * x]}$ and
\[ (q^{[s]})_{\hat x}(\alpha) = q((\hat x * \alpha)^{[s]}) = q_x(\alpha^{[s * x]}) = (q_x)^{[s * x]}(\alpha). \]
\end{proof}

\begin{theorem}[$\HAomega$] $\isp \geq_T \ISP$.
\end{theorem}
\begin{proof} In order to define $\MBR$ of type $(X_i, R)$ we use $\mbr$ of type $(Y_j, R)$. For $\varepsilon_s \colon (X_j \to R) \to \Pi_k X_{k + j}$, where $s \colon \Pi_{i < j} X_i$, define $\tilde \varepsilon_j \colon (Y_j \to R) \to \Pi_k Y_{k + j}$ as
\begin{itemize}
\item[$(i)$]
$\tilde \varepsilon_j(P^{Y_j \to R})(k) \stackrel{Y_{k + j}}{\fdefin} 
\begin{cases}
\pair{\cTrue, \lambda t^{\Pi_{i < j} X_i} . \varepsilon_t(\lambda x^{X_{j}}. P(\hat x))} & \text{if $k = 0$,} \\[2mm]
\pair{\cFalse, \cZero^{\Pi_{i < j + k} X_i \to \Pi_i X_{i + j + k}}} & \text{if $k > 0$}.
\end{cases}$
\end{itemize}
By the definition of $(\cdot)^{[s]}$ and definition $(i)$ it is easy to check that 
\begin{itemize}
	\item[$(ii)$] $(\tilde \varepsilon_{|s|}(P))^{[s]} = \varepsilon_s(\lambda x . P(\hat x))$.
\end{itemize}
We claim that $\MBR$ can be defined from $\isp$ as
\begin{itemize}
\item[$(iii)$] $\MBR_s(\varepsilon)(q^{\Pi_i X_{i + |s|} \to R}) \stackrel{\Pi_i X_{i + |s|}}{\fdefin} (\isp_{|s|}(\tilde \varepsilon)(q^{[s]}))^{[s]}$.
\end{itemize}
We have
\eqleft{
\begin{array}{lcl}
\ISP_s(\varepsilon)(q)
	& \stackrel{(iii)}{=} & (\isp_{|s|}(\tilde \varepsilon)(q^{[s]}))^{[s]} \\[1mm]
	& \stackrel{(\ref{mbr-def-eq})}{=} & \big( \tilde \varepsilon_{|s|}(\lambda f^{Y_{|s|}} . (q^{[s]})_f (\isp_{|s| + 1}(\tilde \varepsilon)((q^{[s]})_f))) \big)^{[s]} \\[1mm]
	& \stackrel{(ii)}{=} & \varepsilon_s(\lambda x^{X_{|s|}} . (q^{[s]})_{\hat x} (\isp_{|s| + 1}(\tilde \varepsilon)((q^{[s]})_{\hat x}))) \\[1mm]
	& \stackrel{\textup{L}\ref{isp-ISP-lemma}}{=} & \varepsilon_s (\lambda x^{X_{|s|}}. q_x(\big(\isp_{|s * x|}(\tilde \varepsilon)((q_x)^{[s * x]})\big)^{[s * x]})) \\[1mm]
	& \stackrel{(iii)}{=} & \varepsilon_s (\lambda x^{X_{|s|}}. q_x(\ISP_{s * x}(\varepsilon)(q_x))).
\end{array}
}
\end{proof}

We now show that a slight generalisation of modified bar recursion \cite{BO(02A),BO(05)} is $T$-equivalent to the iterated product of skewed selection functions. Define $\MBR'$ as
\begin{equation} \label{MBR-var-eq} \tag{\MBR'}
\MBR'_s(\varepsilon)(q) \stackrel{\Pi_i X_i}{=} s * \varepsilon_s(\lambda x^{X_{|s|}} . q(\MBR'_{s * x}(\varepsilon)(q)))
\end{equation}
where $q \colon \Pi_i X_i \to R$ and $\varepsilon_s \colon (X_{|s|} \to R) \to \Pi_i X_{|s|+i}$. $\MBR'$ is a generalisation of modified bar recursion (as defined in \cite{BO(02A)}, cf. lemma 2) to sequence types. If all $X_i = X$ we have precisely the definition given in \cite{BO(02A),BO(05)}.

\begin{theorem}[$\HAomega + \BI + \CONT$] \label{mbr} $\MBR =_T \MBR'$.
\end{theorem}
\begin{proof} For one direction, let $q \colon \Pi_i X_i \to R$ and $s \colon \Pi_{i < n} X_i$ and define
\begin{itemize}
\item[$(i)$] $\MBR'_s(\varepsilon)(q) \stackrel{\Pi_i X_i}{\fdefin} s * \ISP_s(\varepsilon)(q_s)$.
\end{itemize}
Unfolding definitions we have
\eqleft{
\begin{array}{lcl}
\MBR'_s(\varepsilon)(q) & \stackrel{(i)}{=} & s * \ISP_s(\varepsilon)(q_s) \\[1mm]
	& \stackrel{(\ref{MBR-def-eq-unfold})}{=} & s * \varepsilon_s(\lambda x^{X_{|s|}} . q_{s*x}(\ISP_{s*x}(\varepsilon)(q_{s*x}))) \\[2mm]
	& = & s * \varepsilon_s(\lambda x^{X_{|s|}} . q(s * x * \ISP_{s*x}(\varepsilon)(q_{s*x}))) \\[1mm]
	& \stackrel{(i)}{=} & s * \varepsilon_s(\lambda x^{X_{|s|}} . q(\MBR'_{s*x}(\varepsilon)(q))). 
\end{array}
}
For the other direction, let $q \colon \Pi_i X_{i + |s|} \to R$. Define
\begin{itemize}
\item[$(ii)$] $\ISP_s(\varepsilon)(q) \stackrel{\Pi_i X_{i + |s|}}{\fdefin} \MBR'_{\pair{\,}}(\lambda t . \varepsilon_{s * t})(q)$.
\end{itemize}
We then have
\eqleft{
\begin{array}{lcl}
\ISP_s(\varepsilon)(q) & \stackrel{(ii)}{=} & \MBR'_{\pair{\,}}(\lambda t . \varepsilon_{s * t})(q) \\[1mm]
	& \stackrel{(\ref{MBR-var-eq})}{=} & \varepsilon_s(\lambda x^{X_{|s|}} . q(\MBR'_{x}(\lambda t . \varepsilon_{s * t})(q))) \\[1mm]
	& \stackrel{(*)}{=} & \varepsilon_s(\lambda x^{X_{|s|}} . q_x(\MBR'_{\pair{\,}}(\lambda t . \varepsilon_{s * x * t})(q_x))) \\[1mm]
	& \stackrel{(ii)}{=} & \varepsilon_s(\lambda x^{X_{|s|}} . q_x(\ISP_{s*x}(\varepsilon)(q_x))),
\end{array}
}
where $(*) \; \MBR'_{x * r}(\lambda t . \varepsilon_{s * t})(q) = x * \MBR'_{r}(\lambda t . \varepsilon_{s * x * t})(q_x)$ can be proven by bar induction on the sequence $r$, assuming continuity of $q$ (cf. Lemma \ref{epq-eps-lemma}).
\end{proof}

\begin{corollary} Gandy's functional $\Gamma$ is $T$-equivalent to $\ISP$ with $X_i = \NN$ for all $i \in \NN$.
\end{corollary}
\begin{proof} It has been shown in \cite{BO(05)} that the $\Gamma$ functional is $T$-equivalent to $\MBR'$ of lowest type. It remains to observe that the equivalence of Theorem \ref{mbr} respects the types. \end{proof}

\begin{question} It should be mentioned that in \cite{Berardi(98)} yet another form of bar recursion is used for the interpretation of the double negation shift (although they also use modified bar recursion when dealing with dependent choice). We refer to this different bar recursion as the $\BBC$ functional. Thomas Powell \cite{Powell(2014A)} has recently shown that $\BBC$ is $T$-equivalent to $\IPS$ (see also \cite{Berger04}).
\end{question}

%%%%%%%%%%%%%%%%%%%%%%%%%%%%%%%%%%%%%%%%
%%%%%%%%%%%%%%%%%%%%%%%%%%%%%%%%%%%%%%%%
\section{Further Inter-definability Results}
%%%%%%%%%%%%%%%%%%%%%%%%%%%%%%%%%%%%%%%%
%%%%%%%%%%%%%%%%%%%%%%%%%%%%%%%%%%%%%%%%
\label{further-def}

In this section we prove three further inter-definability results, namely $\ips \geq \mbr$, $\MBR \geq \IPS$ and $\IPS \geq \EPQ$.

\begin{theorem}[$\HAomega$] $\ips \geq_T \mbr$.
\end{theorem}
\begin{proof} Given a type $X$ let us denote by $X'$ the type $\BB \times X$. In order to define $\mbr$ of type $(X_i, R)$ we use $\ips$ of type $(\Pi_j X_{i + j}', R)$. The main idea for the construction is to turn a skewed selection function into a proper selection function as follows. Given $\varepsilon_i \colon (X_i \to R) \to \Pi_j X_{i+j}$ we define $\tilde{\varepsilon}_i \colon J (\Pi_j X_{i+j}')$ as
\begin{itemize}
\item[$(i)$] $\tilde{\varepsilon}_i(f^{\Pi_j X_{i+j}' \to R}) \stackrel{\Pi_j X_{i+j}'}{\fdefin} \lambda j . \pair{\cFalse, \varepsilon_i(\lambda x^{X_i} . f(\hat x))(j)}$,
\end{itemize}
where
\eqleft{
\hat x(j) \fdefin 
\left\{
\begin{array}{ll}
\pair{\cTrue, x^{X_i}} & {\rm if} \; j = 0 \\[2mm]
\pair{\cTrue, \cZero^{X_{i+j}}} & {\rm if} \; j > 0.
\end{array}
\right.
}
Intuitively, the booleans $\BB = \{ \cTrue, \cFalse \}$ are used to distinguish between values returned by $\varepsilon_i$ and those values $\hat x$ passed into a recursive call. \\[1mm]
Given $\alpha \colon \Pi_k (\Pi_j X'_{j + i + k})$ we define $\tilde \alpha \colon \Pi_j X_{j + i}$ as
\eqleft{\tilde{\alpha}(j) \stackrel{X_{j + i}}{\fdefin} 
\left\{
\begin{array}{ll}
(\alpha(j)(0))_1 & {\rm if} \; \forall k \!<\! j \, ((\alpha(k)(0))_0 = \cTrue \\[2mm]
(\alpha(k)(j-k))_1 & {\rm if} \; \exists k \!<\! j \, ((\alpha(k)(0))_0 = \cFalse,
\end{array}
\right.
}
where $k = \mu k \!<\! j \, (\alpha(k)(0))_0 = \cFalse$. The construction $\tilde \alpha$ receives as input a matrix $\alpha \colon \Pi_i \Pi_{j \geq i} X_j'$ and produces a sequence $\Pi_j X_j$ as follows: As long as the value $\alpha(j)$ is some $\hat x$ (boolean flag will be $\cTrue$) we filter out the $x$; once we reach a value returned by an $\varepsilon_k$ (boolean will be $\cFalse$) then we return the whole sequence returned by the skewed selection function $\varepsilon_k$. Hence, given a $q \colon \Pi_j X_{i+j} \to R$ we define $\tilde q \colon \Pi_k (\Pi_j X'_{i+k+j}) \to R$ as $\tilde q(\alpha) \fdefin q(\tilde \alpha)$ where, Clearly
\begin{itemize}
	\item[$(ii)$] $\tilde{q}_{\hat x}(\beta) = \widetilde{(q_x)}(\beta)$.
\end{itemize}
We claim that $\isp$ can be defined as
\begin{itemize}
\item[$(iii)$] $\isp_i(\varepsilon)(q) \stackrel{\Pi_j X_{i+j}}{\fdefin} (\big(\ips_i(\tilde{\varepsilon})(\tilde{q})\big)^{\Pi_k \Pi_j X'_{i+k+j}}(0))^1$
\end{itemize}
where $\varepsilon_i \colon (X_i \to R) \to \Pi_j X_{i+j}$ and $q \colon \Pi_j X_{i+j} \to R$.  Recall that given a sequence $\beta \colon \Pi_{i < n} (X_i \times Y_i)$ we write $\beta^1  \colon \Pi_{i < n} Y_i$ for the projection of the sequence on the second coordinates. We have
\eqleft{
\begin{array}{lcl}
\isp_i(\varepsilon)(q) 
	& \stackrel{(iii)}{=} & (\ips_i(\tilde{\varepsilon})(\tilde{q})(0))^1 \\[1mm]
	& \stackrel{(\ref{ips-def-eq})}{=} & (\tilde{\varepsilon}_i(\lambda \alpha^{\Pi_j X'_{i+j}} . \tilde{q}_{\alpha}(\ips_{i+1}(\tilde{\varepsilon})(\tilde{q}_{\alpha}))))^1 \\[1mm]
	& \stackrel{(i)}{=} & \varepsilon_i(\lambda x^{X_i} . \tilde{q}_{\hat x}(\ips_{i+1}(\tilde{\varepsilon})(\tilde{q}_{\hat x}))) \\[1mm]
	& \stackrel{(ii)}{=} & \varepsilon_i(\lambda x^{X_i} . \widetilde{(q_x)}(\ips_{i+1}(\tilde{\varepsilon})(\widetilde{(q_x)}))) \\[1mm]
	& \stackrel{(iv)}{=} & \varepsilon_i(\lambda x^{X_i} . q_x((\ips_{i+1}(\tilde{\varepsilon})(\widetilde{(q_x)})(0))^1)) \\[1mm]
	& \stackrel{(iii)}{=} & \varepsilon_i(\lambda x^{X_i} . q_x(\mbr_{i+1}(\varepsilon)(q_x)))
\end{array}
}
using that
\begin{itemize}
	\item[$(iv)$] $\tilde \beta = (\ips_{i+1}(\tilde{\varepsilon})(\widetilde{(q_x)})(0))^1$, for $\beta = \ips_{i+1}(\tilde{\varepsilon})(\widetilde{(q_x)})$.
\end{itemize}
\end{proof}

%%%%%%%%%%%%%%%%%%%%%%%%%%%%%%%%%%%%%%%%
\subsection{$\MBR \geq \IPS$}
%%%%%%%%%%%%%%%%%%%%%%%%%%%%%%%%%%%%%%%%
\label{mbr-def-ips}

We now show that the implicitly controlled dependent product of selection functions $\IPS$ is $T$-definable from (and hence $T$-equivalent to) modified bar recursion $\MBR$. Since in our proof we need to work with infinite sequences of finite sequences (of arbitrary length), the use of the infinite product type $\Pi_i X_i$ in here is unhelpful. The problem is that keeping track of indices would imply introducing the $\Sigma$ type to record the length of each finite sequence. Although this can be done, it would require much more of dependent type theory than we have assumed so far. Hence, for this section only we work with selection functions of a fixed type  $(X \to R) \to X$. Similarly, skewed selection functions will have type $(Y \to R) \to Y^\NN$. 

Let $X^+$ denote non-empty finite sequences of elements of type $X$. We make use of the following two mappings $G \colon X \to X^+$ and $F \colon (X^+)^\NN \to X^\NN$ where
\eqleft{
\begin{array}{lcl}
	G(x) & = & \langle x \rangle \\[1mm]
	F(\alpha) & = & \mbox{concatenation of non-empty finite sequences $\alpha(i)$'s.}
\end{array}
}
For the definition of $F$ it is important that we are considering non-empty sequences, as otherwise such concatenation operation would not be defined in general. Consider two variants $G^* \colon X^* \to (X^+)^*$ and $F^* \colon (X^+)^* \to X^*$, where $G^*$ is  the function $G$ applied pointwise to a given finite sequence, and $F^*$ is the concatenation of a finite sequence of non-empty finite sequences. 

\begin{lemma} \label{mbr-ips-lemmaA} The following can be easily verified:
\begin{itemize}
	\item[($i$)] $F(\lambda i . G(v_i))(i) = v_i$ and $F^*(G^*(s)) = s$, where $v_i \colon X$ and $s \colon X^*$. \\[-2mm]
	\item[($ii$)] $F^*(s * t) = F^*(s) * F^*(t)$, where $s, t \colon (X^+)^*$. \\[-2mm]
	\item[($iii$)] $F^*(G^*(s) * t) = s * F^*(t)$, where $s \colon X^*$ and $t \colon (X^+)^*$.
\end{itemize}
\end{lemma}

Given selection functions $\varepsilon_s \colon J_R X$ define, by course-of-values, skewed selection functions of type 
\eqleft{\nu_r \colon (X^+ \to R) \to (X^+)^\NN,}
where $r \colon (X^+)^*$, as
\eqleft{\nu_r(P^{X^+ \to R})(i) \stackrel{X^+}{=} G(\varepsilon_{F^*(r * t^i)}(\lambda x^X . P(\langle (F^* t^i) * x \rangle)))}
with $t^i \stackrel{(X^+)^*}{=} \initSeg{\nu_r(P^{X^+ \to R})}{i}$. 

\begin{lemma} \label{mbr-ips-lemmaB} If $F^*(r) = F^*(r')$ then $\nu_{r}(P)(i) = \nu_{r'}(P)(i)$.
\end{lemma}

\begin{proof} Directly from Lemma \ref{mbr-ips-lemmaA} $(ii)$ since $F^*(r * t^i) = F^*(r) * F^*(t^i)$.
%Assume $(\dagger) \, F^*(r) = F^*(r')$. By course-of-values induction we have
%%
%\eqleft{
%\begin{array}{lcl}
%\nu_r(P)(i) 
%	& = & G(\varepsilon_{F^*(r * t^i)}(\lambda x . P(\langle (F^* t^i) * x \rangle))) \\[1mm]
%	& \stackrel{\textup{L}\ref{mbr-ips-lemmaA}(ii)}{=} & G(\varepsilon_{F^*(r) * F^*(t^i)}(\lambda x . P(\langle (F^* t^i) * x \rangle))) \\[1mm]
%	& \stackrel{(\dagger)}{=} & G(\varepsilon_{F^*(r') * F^*(t^i)}(\lambda x . P(\langle (F^* t^i) * x \rangle))) \\[1mm]
%	& \stackrel{(\textup{IH})}{=} & G(\varepsilon_{F^*(r') * F^*(u^i)}(\lambda x . P(\langle (F^* u^i) * x \rangle))) \\[1mm]
%	& \stackrel{\textup{L}\ref{mbr-ips-lemmaA}(ii)}{=} & G(\varepsilon_{F^*(r' * u^i)}(\lambda x . P(\langle (F^* u^i) * x \rangle))) \\[2mm]
%	& = & \nu_{r'}(P)(i)
%\end{array}
%}
%%
%where $t^i = \initSeg{\nu_r(P)}{i}$ and $u^i = \initSeg{\nu_{r'}(P)}{i}$.
\end{proof}

Now, given a functional $q \colon X^\NN \to R$ define $\tilde q \colon (X^+)^\NN \to R$ as
\eqleft{\tilde q (\alpha) \stackrel{R}{=} q(F \alpha).}
Again, it is easy to see that:

\begin{lemma} \label{mbr-ips-lemmaC} $(\tilde q)_{\langle s \rangle} \stackrel{(X^+)^\NN}{=} \widetilde{(q_s)}$, where $s \colon X^+$.
\end{lemma}
\begin{proof} $(\tilde q)_{\langle s \rangle}(\alpha) = q(F(\langle s \rangle * \alpha)) = q(F^*(\langle s \rangle) * F(\alpha))) = q_s(F \alpha) = \widetilde{(q_s)}(\alpha).$
\end{proof}

\begin{lemma}[$\HAomega + \BI + \CONT$] \label{mbr-ips-lemmaD} If $F^*(s) = F^*(s')$ then
\[ \MBR_{s}(\nu)(\tilde q) = \MBR_{s'}(\nu)(\tilde q). \]
\end{lemma}
\begin{proof} Define the predicate
\eqleft{P(t^{(X^+)^*}) \,\equiv\, F^*(s * t) = F^*(s' * t) \to \MBR_{s * t}(\nu)(\tilde{q}_t) = \MBR_{s' * t}(\nu)(\tilde{q}_t).}
We show $P(\langle \, \rangle)$ by bar induction $\BI$ (assuming $\CONT$). \\[1mm]
$(i)$ $\forall \alpha \exists k P(\initSeg{\alpha}{k})$. Given $\alpha$, by $\CONT$ let $k$ be such that $\tilde{q}_{\initSeg{\alpha}{k}}$ is a constant function, say $\tilde{q}_{\initSeg{\alpha}{k}}(\beta) = r$, for all $\beta$. Assuming $F^*(s * \initSeg{\alpha}{k}) = F^*(s' * \initSeg{\alpha}{k})$
\[
\begin{array}{lcl}
\MBR_{s * \initSeg{\alpha}{k}}(\nu)(\tilde{q}_{\initSeg{\alpha}{k}})
	& = & \nu_{s * \initSeg{\alpha}{k}}(\lambda y . \tilde{q}_{\initSeg{\alpha}{k} * y}(\MBR_{s * \initSeg{\alpha}{k} * y}(\nu)(\tilde{q}_{\initSeg{\alpha}{k} * y}))) \\[2mm]
	& = & \nu_{s * \initSeg{\alpha}{k}}(\lambda y . r) \\[1mm]
	& \stackrel{\textup{L}\ref{mbr-ips-lemmaB}}{=} & \nu_{s' * \initSeg{\alpha}{k}}(\lambda y . r) \\[2mm]
	& = & \nu_{s' * \initSeg{\alpha}{k}}(\lambda y . \tilde{q}_{\initSeg{\alpha}{k} * y}(\MBR_{s' * \initSeg{\alpha}{k} * y}(\nu)(\tilde{q}_{\initSeg{\alpha}{k} * y}))) \\[2mm]
	& = & \MBR_{s' * \initSeg{\alpha}{k}}(\nu)(\tilde{q}_{\initSeg{\alpha}{k}}).
\end{array}
\]
$(ii)$ $\forall t (\forall x P(t * x) \to P(t))$. Let $t$ be such that $\forall x P(t * x)$. Assuming $F^*(s * t) = F^*(s' * t)$, and noting that this implies $F^*(s * t * y) = F^*(s' * t * y)$, we have
\eqleft{
\begin{array}{lcl}
\MBR_{s * t}(\nu)(\tilde{q}_t)
	& = & \nu_{s * t}(\lambda y . \tilde{q}_{t * y}(\MBR_{s * t * y}(\nu)(\tilde{q}_{t * y}))) \\[1mm]
	& \stackrel{(\textup{IH})}{=} & \nu_{s * t}(\lambda y . \tilde{q}_{t * y}(\MBR_{s' * t * y}(\nu)(\tilde{q}_{t * y}))) \\[1mm]
	& \stackrel{\textup{L}\ref{mbr-ips-lemmaB}}{=} & \nu_{s' * t}(\lambda y . \tilde{q}_{t * y}(\MBR_{s' * t * y}(\nu)(\tilde{q}_{t * y}))) \\[2mm]
	& = & \MBR_{s' * t}(\nu)(\tilde{q}_t).
\end{array}
}
\end{proof}

We can now show that $\IPS$ of type $(X, R)$ is $T$-definable from $\MBR$ of type $(X^+, R)$.

\begin{theorem}[$\HAomega + \BI + \CONT$] \label{mbr-cbr} $\MBR \geq_T \IPS$.
\end{theorem}
\begin{proof}  Define $\IPS$ from $\MBR$ as
\eqleft{\IPS_s(\varepsilon)(q) \stackrel{X^\NN}{=} F(\MBR_{G^*(s)}(\nu)(\tilde q))}
where $\nu$ (and $\tilde q$) is defined from $\varepsilon$ ($q$, respectively) as above. We show that $\IPS$ as defined above satisfies its defining equation.
Let
\begin{itemize}
	\item $t^i = \initSeg{\nu_{G^*s}(\lambda y . \tilde{q}_y(\MBR_{(G^*s) * y}(\nu)(\tilde{q}_y)))}{i}$
	\item $r^i = \initSeg{\IPS_s(\varepsilon)(q)}{i}$. 
\end{itemize}
We first show that $(\dagger)~F^*(t^i) = r^i$. By course-of-values assume $F^*(t^j) = r^j$ for $j < i$, then
\[
\begin{array}{lcl}
F^*(t^i)(j)
	& \stackrel{X^+}{=} & F^*(\initSeg{\nu_{G^*s}(\lambda y . \tilde{q}_y(\MBR_{(G^*s) * y}(\nu)(\tilde{q}_y)))}{i})(j) \\[1mm]
	& \stackrel{\textup{L}\ref{mbr-ips-lemmaA}(i)}{=} & \varepsilon_{F^*((G^* s) * t^j)}(\lambda x . \tilde{q}_{\langle (F^* t^j) * x \rangle}(\MBR_{(G^* s) * \langle (F^* t^j) * x \rangle}(\nu)(\tilde{q}_{\langle (F^* t^j) * x \rangle}))) \\[1mm]
	& \stackrel{\textup{L}\ref{mbr-ips-lemmaA}(iii)}{=} & \varepsilon_{s * F^*(t^j)}(\lambda x . \tilde{q}_{\langle (F^* t^j) * x \rangle}(\MBR_{(G^* s) * \langle (F^* t^j) * x \rangle}(\nu)(\tilde{q}_{\langle (F^* t^j) * x \rangle}))) \\[1mm]
	& \stackrel{\textup{L}\ref{mbr-ips-lemmaC}}{=} & \varepsilon_{s * F^*(t^j)}(\lambda x . \widetilde{q_{(F^* t^j) * x}}(\MBR_{(G^* s) * \langle (F^* t^j) * x \rangle}(\nu)(\widetilde{q_{(F^* t^j) * x}}))) \\[1mm]
	& \stackrel{\textup{(IH)}}{=} & \varepsilon_{s * r^j}(\lambda x . \widetilde{q_{r^j * x}}(\MBR_{(G^* s) * \langle r^j * x \rangle}(\nu)(\widetilde{q_{r^j * x}}))) \\[1mm]
	& \stackrel{\textup{L}\ref{mbr-ips-lemmaD}}{=} & \varepsilon_{s * r^j}(\lambda x . \widetilde{q_{r^j * x}}(\MBR_{G^*( s * r^j * x)}(\nu)(\widetilde{q_{r^j * x}}))) \\[2mm]
	& = & \varepsilon_{s * r^j}(\lambda x . q_{r^j * x}(F(\MBR_{G^*(s * r^j * x)}(\nu)(\widetilde{q_{r^j * x}})))) \\[2mm]
	& = & \varepsilon_{s * r^j}(\lambda x . q_{r^j * x}(\IPS_{s * r^j * x}(\varepsilon)(q_{r^j * x}))) \\[2mm]
	& = & (r^i)(j).
\end{array}
\]
We then have 
\[
\begin{array}{lcl}
\IPS_s(\varepsilon)(q)(i) 
	& \stackrel{X}{=} & F(\MBR_{G^*(s)}(\nu)(\tilde q))(i) \\[2mm]
	& = & F(\nu_{G^* s} (\lambda y^{X^+} . \tilde{q}_y(\MBR_{(G^*s) * y}(\nu)(\tilde{q}_y))))(i) \\[1mm]
	& \stackrel{\textup{L}\ref{mbr-ips-lemmaA}(i)}{=} & \varepsilon_{F^*((G^* s) * t^i)} (\lambda x . \tilde{q}_{\langle (F^* t^i) * x \rangle}(\MBR_{(G^* s) * \langle (F^* t^i) * x \rangle}(\nu)(\tilde{q}_{\langle (F^* t^i) * x \rangle})))) \\[1mm]
	& \stackrel{\textup{L}\ref{mbr-ips-lemmaA}(ii)}{=} & \varepsilon_{s * F^*(t^i)} (\lambda x . \tilde{q}_{\langle (F^* t^i) * x \rangle}(\MBR_{(G^* s) * \langle (F^* t^i) * x \rangle}(\nu)(\tilde{q}_{\langle (F^* t^i) * x \rangle})))) \\[1mm]
	& \stackrel{(\dagger)}{=} & \varepsilon_{s * r^i} (\lambda x . \tilde{q}_{\langle r^i * x \rangle}(\MBR_{(G^* s) * \langle r^i * x \rangle}(\nu)(\tilde{q}_{\langle r^i * x \rangle})))) \\[1mm]
	& \stackrel{\textup{L}\ref{mbr-ips-lemmaC}}{=} & \varepsilon_{s * r^i} (\lambda x . \widetilde{q_{r^i * x}}(\MBR_{(G^* s) * \langle r^i * x \rangle}(\nu)(\widetilde{q_{r^i * x}})))) \\[2mm]
	& = & \varepsilon_{s * r^i} (\lambda x . q_{r^i * x}(F(\MBR_{(G^* s) * \langle r^i * x \rangle}(\nu)(\widetilde{q_{r^i * x}})))) \\[1mm]
	& \stackrel{\textup{L}\ref{mbr-ips-lemmaD}}{=} & \varepsilon_{s * r^i} (\lambda x . q_{r^i * x}(F(\MBR_{G^* (s * r^i * x)}(\nu)(\widetilde{q_{r^i * x}})))) \\[2mm]
	& = & \varepsilon_{s * r^i} (\lambda x . q_{r^i * x}(\IPS_{s * r^i * x}(\varepsilon)(q_{r^i * x}))).
\end{array}
\]
\end{proof}

%%%%%%%%%%%%%%%%%%%%%%%%%%%%%%%%%%%%%%%%
\subsection{$\IPS \geq \EPQ$}
%%%%%%%%%%%%%%%%%%%%%%%%%%%%%%%%%%%%%%%%
\label{ips-def-epq}

It has been shown in \cite{BO(05)} that $\BR$ is $T$-definable from modified bar recursion. Here we simplify that construction and use it to show that $\EPQ$ is $T$-definable from $\IPS$. Moreover, we make explicit the assumption $\SPEC$ which is used in \cite{BO(05)}. First we prove that (the totalisation of) \emph{Spector's search functional} is definable in G\"odel's system $T$.

\begin{lemma}[$\HAomega$] \label{spec-search} The totalisation of \emph{Spector's search functional}
\eqleft{\mu_{{\sf sc}}(\omega)(\alpha) = {\sf least} \, n (\omega(\initSegZ{\alpha}{n}) < n)}
is $T$-definable. More precisely, there exists a term $\chi$ in G\"odel's system $T$ such that the following is provable in $\HAomega$
\eqleft{\exists n (\omega(\initSegZ{\alpha}{n}) < n) \to (\omega(\initSegZ{\alpha}{\chi \omega \alpha}) < \chi \omega \alpha \wedge \forall i < \chi \omega \alpha (\omega(\initSegZ{\alpha}{i}) \geq i)).}
In particular,
\eqleft{\HAomega + \SPEC \vdash (\omega(\initSegZ{\alpha}{n}) < n) \wedge \forall i < n (\omega(\initSegZ{\alpha}{i}) \geq i)}
where $n = \chi \omega \alpha$.
\end{lemma}
\begin{proof} We show how the unbounded search in $\mu_{{\sf sc}}$ can be turned into a bounded search. Abbreviate $A_n(\omega, \alpha) \fdefin (\omega(\initSegZ{\alpha}{n}) < n)$. Consider the following construction, given $\alpha \colon \Pi_i X_i$ define $\alpha^\omega \colon \Pi_i X_i$ as
\eqleft{\alpha^\omega(i) = 
\left\{
\begin{array}{ll}
	\cZero^{X_i} & {\rm if} \; \exists k \!\leq\! i+1 \, A_k(\omega, \alpha) \\[2mm]
	\alpha(i) \quad & {\rm otherwise}.
\end{array}
\right.
}
Assume $\exists n (\omega(\initSegZ{\alpha}{n}) < n)$. Let $n$ be the least number such that $A_n(\omega, \alpha)$ holds. Then it is easy to see that $\alpha^\omega = \initSegZ{\alpha}{n-1}$. Because $n$ is least, we must have that $\omega(\alpha^\omega) \geq n-1$, and hence $n \leq \omega(\alpha^\omega) + 1$. Therefore, $\omega(\alpha^\omega) + 1$ serves as an upper bound on the search $\mu_{\sf sc}$, i.e. $\chi \omega \alpha = \mu n \leq \omega(\alpha^\omega) + 1 \; (\omega(\initSegZ{\alpha}{n}) < n)$. \end{proof}

The construction above shows that Spector's search functional can be made total in system $T$, so that whenever it is well-defined for inputs $\omega$ and $\alpha$ the term $\chi$ indeed computes the correct value.

\begin{theorem} \label{IPS-EPQ-def} $\EPQ$ is $T$-definable from $\IPS$ over $\HAomega + \SPEC$. However, $\EPQ$ is not $T$-definable from $\IPS$, even over $\HAomega + \BI + \CONT$.
\end{theorem}
\begin{proof} First, note that combining the results above we have the equivalences $\IPS =_T \MBR$ and $\EPQ =_T \BR$, over $\HAomega + \BI + \CONT$. Hence, that $\IPS$ is not $T$-definable from $\EPQ$, even over $\HAomega + \BI + \CONT$, follows from fact that $\MBR$ is not S1-S9 computable in the model of total continuous functions while $\BR$ is (see \cite{BO(05)}). \\[1mm]
In order to show $\IPS \geq_T \EPQ$ we use the search operator $\chi$ of the above lemma. Define
\eqleft{\chi^{+k}(\omega)(\alpha) = \mu i \leq \chi(\lambda \beta . \omega(\beta) - k)(\alpha) \; (\omega(\initSeg{\alpha}{i}) < i + k) ,}
where $\omega(\beta) - k$ is the cut-off subtraction. By Lemma \ref{spec-search} we have that $n = \chi^{+k}(\omega)(\alpha)$ is the least such that $\omega(\initSeg{\alpha}{n}) - k < n$. But since $\omega(\initSeg{\alpha}{i}) - k < i$ implies $\omega(\initSeg{\alpha}{i}) < i + k$, we have that, provably in $\HAomega + \SPEC$,
\begin{itemize}
	\item[$(i)$] $\omega(\initSeg{\alpha}{n}) < n + k$ and $\forall i < n (\omega(\initSeg{\alpha}{i}) \geq i + k)$, for $n = \chi^{+k}(\omega)(\alpha)$.
\end{itemize}
Let $\psi_s \colon K_R X_{|s|}$ be a given family of quantifiers. We first turn each quantifier $\psi_s \colon K_R X_{|s|}$, where $s \colon \Pi_{i < |s|} X_i$, into a a selection function $\tilde{\psi}_t$ of type $J_R (X_{|t|} \uplus R)$ as\footnote{We are here making use of the sum type $X \uplus Y$, which can be implemented as $\BB \times X \times Y$, since we assume all types are inhabited, with $\inj_X \colon X \to X \uplus Y$ and $\inj_Y \colon Y \to X \uplus Y$ the standard injections.}
\begin{itemize}
	\item[$(ii)$] $\tilde{\psi}_t(F^{(X_{|t|} \uplus R) \to R}) \stackrel{X_{|t|} \uplus R}{\fdefin} \inj_R (\psi_{\check{t}}(\lambda x^{X_{|t|}} . F(\inj_{X_{|t|}} x)))$
\end{itemize}
where $t \colon \Pi_{i < |t|} (X_i \uplus R)$, and $\check{(\cdot)} \colon \Pi_{i < n} (X_i \uplus R) \to \Pi_{i < n} X_i$ is defined as
\eqleft{(\check{s})_i \stackrel{X_i}{\fdefin} 
\left\{
\begin{array}{ll}
	x_i & {\rm if} \; s_i = \inj_{X_i}(x_i) \\[2mm]
	\cZero^{X_i} \quad & {\rm otherwise}.
\end{array}
\right.
}
We will also make use of the dual operation $\tilde{(\cdot)} \colon \Pi_{i < n} X_i \to \Pi_{i < n} (X_i \uplus R)$ that maps $\inj_X(\cdot)$ pointwise on a given sequence. Clearly we have
\begin{itemize}
	\item[$(iii)$] $\check{\tilde{s}} = s$ and $\tilde s * \inj_{X_{|s|}}(x) = \widetilde{s*x}$, for $s \colon \Pi_{i < n} X_i$.
\end{itemize}
Note that both construction $\check{(\cdot)}$ and $\tilde{(\cdot)}$ can similarly defined on infinite sequences as well. Hence, given $s \colon \Pi_{i < k} X_i$ and $q \colon \Pi_i X_{i + k} \to R$ and $\length \colon R \to \NN$ let us define the function $q^{\length, s} \colon \Pi_{i \geq k} (X_i \uplus R) \to R$ as
\begin{itemize}
	\item[$(iv)$] $q^{\length, s}(\alpha^{\Pi_i(X_{i + k} \uplus R)}) \stackrel{R}{\fdefin} 
\left\{
\begin{array}{ll}
	q(\initSeg{\check{\alpha}}{n} * \cZero) & {\rm if} \; \forall i \!<\! n \, (\alpha(i) \in X_{i + k}) \\[2mm]
	a & {\rm if} \; \exists i \!<\! n \, (\alpha(i) \in R),
\end{array}
\right.$
\end{itemize}
where $n \fdefin \chi^{+|s|}(\length \circ q)(\check{\alpha})$ and $\alpha(\mu i \!<\! n \, (\alpha(i) \in R)) = \inj_R(a)$. Intuitively, when $q^{\length, s}$ reads an input sequence $\alpha \colon \Pi_i (X_{i + |s|} \uplus R)$ it finds the first point $n$ where $\length(q(\initSegZ{\check{\alpha}}{n})) < n + |s|$. If all values in $\alpha$ up to that point are $X_i$ values it means this $\alpha$ was generated by a sequence of bar recursive calls until the stopping condition was reached, and hence we must apply the outcome function $q$ to the sequence up to that point. Otherwise, it means that the bar recursive calls have already reached the leaves of the bar recursion (i.e. the stopping conditions) and we are now backtracking and calculating the values of intermediate notes, i.e. computations of the $R$-values. In which case the first such value is then returned. We claim that $\EPQ$ defined as
\begin{itemize}
	\item[$(v)$] $\EPQ_s^{\length}(\psi)(q^{\Pi_i X_{i + |s|} \to R}) \stackrel{R}{\fdefin} q^{\length, s}(\IPS_{\tilde s}(\tilde{\psi})(q^{\length, s}))$
\end{itemize}
satisfies equation $(\ref{EPQ-def})$. Consider two cases. \\[1mm]
If $\length(q(\cZero)) < |s|$ then, by $(i)$, $n = \chi^{+|s|}(\length \circ q)(\beta) = 0$, for any $\beta$. Hence, by $(iv)$, we have that $q^{\length, s}(\beta) = q(\cZero)$, again for any $\beta$. Therefore,
\eqleft{
\begin{array}{lcl}
\EPQ_s^{\length}(\psi)(q)
	& \stackrel{(v)}{\fdefin} & q^{\length, s}(\IPS_{\tilde s}(\tilde{\psi})(q^{\length, s})) \\[2mm]
	& = & q(\cZero).
\end{array}
}
On the other hand, if $\length(q(\cZero)) \geq |s|$ then, again by $(i)$, $n = \chi^{+|s|}(\length \circ q)(\beta) > 0$, for any $\beta$. This implies both
\begin{itemize}
	\item[($vi$)] $q^{\length, s}(c * \beta) = r$, for $c = \inj_R(r)$ and arbitrary $\beta$, and \\[-2mm]
	\item[($vii$)] $q^{\length, s}(d * \beta) = (q_x)^{\length, s*x}(\beta)$, for $d = \inj_{X_{|s|}}(x)$ and arbitrary $\beta$.
\end{itemize}
Hence
\eqleft{
\begin{array}{lcl}
\EPQ_s^{\length}(\psi)(q)
	& \stackrel{(v)}{\fdefin} & q^{\length, s}(\IPS_{\tilde s}(\tilde{\psi})(q^{\length, s})) \\[1mm]
	& \stackrel{(\ref{IPS-def})}{\fdefin} & q^{\length, s}(c * \IPS_{\tilde s * c}(\tilde{\psi})((q^{\length, s})_c)) \\[1mm]
	& \stackrel{(ii), (vi)}{=} & \psi_{\check{\tilde{s}}}(\lambda x . (q^{\length, s})_{\inj_{X_{|s|}}(x)}(\IPS_{\tilde s * \inj_{X_{|s|}}(x)}(\tilde \psi)((q^{\length, s})_{\inj_{X_{|s|}(x)}}))) \\[1mm]
	& \stackrel{(iii)}{=} & \psi_s(\lambda x . (q^{\length, s})_{\inj_{X_{|s|}}(x)}(\IPS_{\widetilde{s*x}}(\tilde \psi)((q^{\length, s})_{\inj_{X_{|s|}(x)}}))) \\[1mm]
	& \stackrel{(vii)}{=} & \psi_s(\lambda x . (q_x)^{\length, s*x}(\IPS_{\widetilde{s*x}}(\tilde \psi)((q_x)^{\length, s*x}))) \\[1mm]
	& \stackrel{(v)}{\fdefin} & \psi_s(\lambda x . \EPQ_{s*x}^{\length}(\psi)(q_x)) \\[2mm]
\end{array}
}
where $c = \tilde{\psi}_{\tilde s}(\lambda x . \selEmb{\IPS_{\tilde s * x}(\tilde{\psi})}((q^{\length, s})_x))$. 
% Note that we need to use the assumption $\SPEC$ to prove the properties of $q^{l,s}$ used above, i.e. $(v)$ and $(vi)$, as they only hold assuming $\chi$ indeed behaves like the real Spector's search operator (relying on Lemma \ref{spec-search}).
\end{proof}

\begin{remark} As shown in \cite{Howard(1968)} (cf. also Lemma \ref{howard-kreisel}), if one extends system $T$ with Spector's bar recursion, one can actually prove $\SPEC$. Hence, the result above says that in all models of system $T$ where $\EPQ$ could exist, it indeed does whenever $\IPS$ also exists. We leave it as an open question whether $\IPS$ already defines $\EPQ$ without assuming $\SPEC$.
\end{remark}

\begin{figure}
\includegraphics[width=11.5cm]{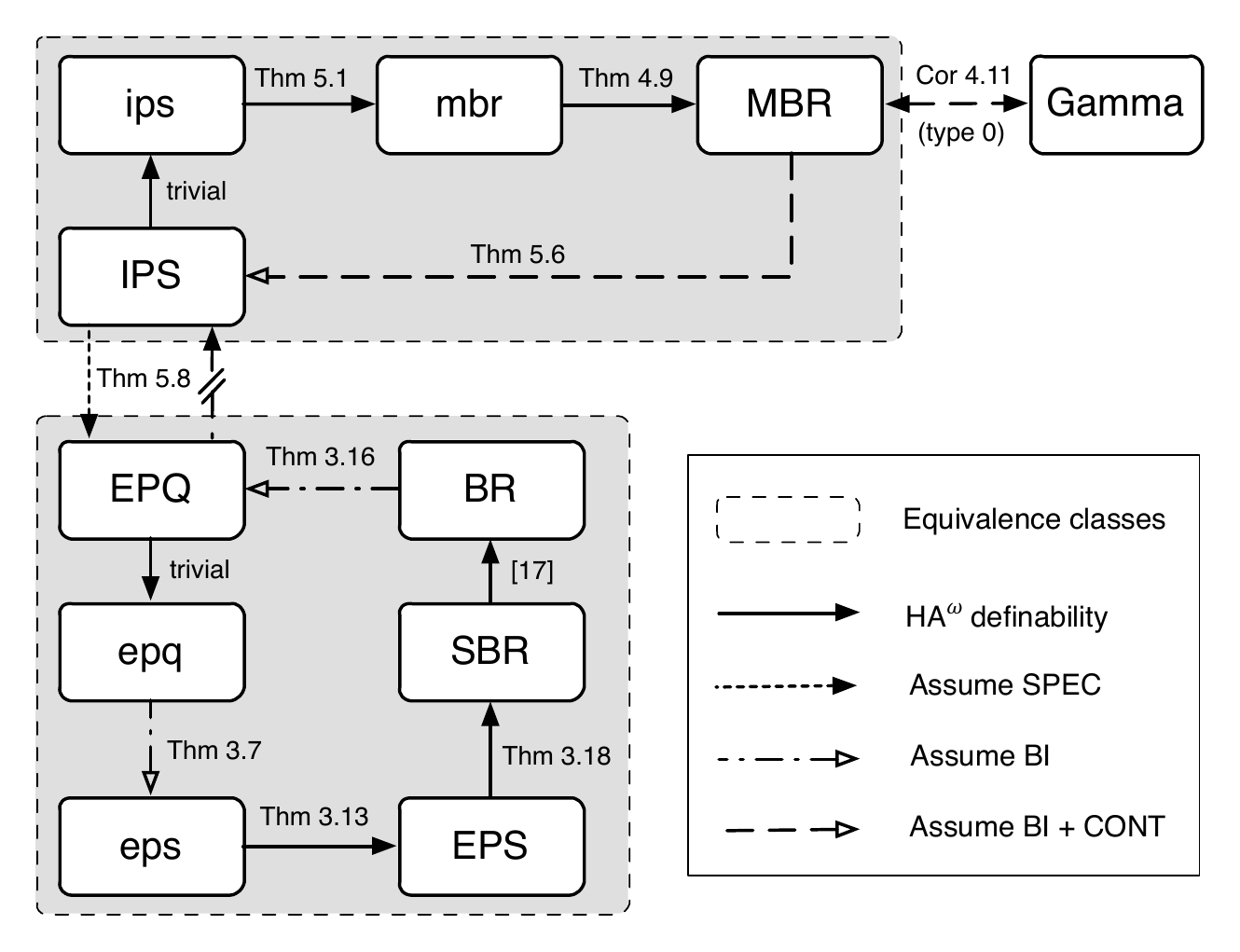}
\caption{Diagram of inter-definability results}
\label{table}
\end{figure}

%%%%%%%%%%%%%%%%%%%%%%%%%%%%%%%%%%%%%%%%
%%%%%%%%%%%%%%%%%%%%%%%%%%%%%%%%%%%%%%%%
\section{Summary of Results}
%%%%%%%%%%%%%%%%%%%%%%%%%%%%%%%%%%%%%%%%
%%%%%%%%%%%%%%%%%%%%%%%%%%%%%%%%%%%%%%%%

Figure \ref{table} gives a diagrammatic picture of the results presented above. We use a full-line-arrow to represent that the inter-definability holds over $\HAomega$, whereas a dotted-line-arrow indicates that extra assumptions are needed. We have used extra assumptions in four cases. In Theorems \ref{cps-sbr} and \ref{BR-EPQ-def} we made use of bar induction $\BI$; in Theorem \ref{IPS-EPQ-def} we use $\SPEC$; and in Theorem \ref{mbr-cbr} we seem to need both bar induction $\BI$ and the axiom of continuity $\CONT$. It is an interesting open question whether any of these four results can be shown in $\HAomega$ alone, or under weaker assumptions.

Given that $\CONT$ implies $\SPEC$, our results show that over the theory $\HAomega + \BI + \CONT$ the different forms of bar recursion considered here fall into two distinct equivalence classes with respect to $T$-definability.

\medskip

\noindent {\bf Acknowledgements}. The authors would like to thank Ulrich Berger, Thomas Powell and in particular the anonymous referee for suggesting numerous improvements and spotting  inaccuracies in earlier versions of the paper. The second author also acknowledges support of The Royal Society under grant 516002.K501/RH/kk.

\bibliographystyle{asl}

\bibliography{dblogic}

\end{document}